\documentclass[10pt,journal]{IEEEtran}
\usepackage{amsmath,amsfonts,amssymb}
\usepackage{amsthm} 
\usepackage{algorithm,algorithmic}
\usepackage{array}
\usepackage[caption=false,font=normalsize,labelfont=sf,textfont=sf]{subfig}
\usepackage{textcomp}
\usepackage{stfloats}
\usepackage{url}
\usepackage{verbatim}
\usepackage{graphicx}
\def\BibTeX{{\rm B\kern-.05em{\sc i\kern-.025em b}\kern-.08em
    T\kern-.1667em\lower.7ex\hbox{E}\kern-.125emX}}
\usepackage{balance}
\usepackage{color}

\usepackage{diagbox}
\usepackage{multirow}  

\definecolor{mygreen}{rgb}{0, 0.65, 0.3}

\newtheorem{proposition}{Proposition}

\newtheorem{lemma}{Lemma}

\begin{document}
\title{Energy Efficiency Optimization in Active Reconfigurable Intelligent Surface-Aided Integrated Sensing and Communication Systems}
\author{Junjie~Ye, \IEEEmembership{Graduate Student Member,~IEEE,}  Mohamed~Rihan, \IEEEmembership{Senior Member,~IEEE},  Peichang~Zhang,  Lei~Huang, \IEEEmembership{Senior Member,~IEEE,}  Stefano~Buzzi, \IEEEmembership{Senior Member,~IEEE,}  Zhen~Chen, \IEEEmembership{Senior Member,~IEEE}
\thanks{Copyright (c) 20xx IEEE. Personal use of this material is permitted. However, permission to use this material for any other purposes must be obtained from the IEEE by sending a request to pubs-permissions@ieee.org. }
\thanks{Junjie Ye, Peichang Zhang, Lei Huang and Zhen Chen are with State Key Laboratory of Radio Frequency Heterogeneous Integration, Shenzhen University, Shenzhen, China; M. Rihan is with Department of Communications Engineering, University of Bremen, Bremen, Germany and with Department of Electronics and Electrical Communication Engineering, Faculty of Electronic Engineering, Menouf, Egypt; S. Buzzi is with the Department of Electrical and Information Engineering, University of Cassino and Southern Lazio, Cassino, Italy, and with Dipartimento di Elettronica, Informazionee Bioingegneria, Politecnico di Milano, Milan, Italy. (e-mail: 2152432003@email.szu.edu.cn;  mohamed.elmelegy@el-eng.menofia.edu.eg; pzhang@szu.edu.cn; lhuang@szu.edu.cn; buzzi@unicas.it; chenz.scut@gmail.com.) (Corresponding author: Lei Huang; Peichang Zhang.)}
\thanks{This work is supported in part by the Key Project of International Cooperation and Exchanges of the National Natural Science Foundation of China under Grant 62220106009, the National Science Fund for Distinguished Young Scholars under Grant 61925108,  the project of Shenzhen Peacock Plan Teams under Grant KQTD20210811090051046 and Research Team Cultivation Program of Shenzhen University under Grant 2023DFT003. The work of Peichang Zhang is supported in part by the Foundation of Shenzhen under Grant 20220810142731001, 20200823154213001. The work of Stefano Buzzi  is supported by the European Union under the Italian National Recovery and Resilience Plan (NRRP) of NextGenerationEU, partnership on “Telecommunications of the Future” (PE00000001 - program “RESTART”, Structural Project SRE). }}

\maketitle
\begin{abstract}
Energy efficiency (EE) is a challenging task in integrated sensing and communication (ISAC) systems, where high spectral efficiency and low energy consumption appear as conflicting requirements. Although passive reconfigurable intelligent surface (RIS) has emerged as a promising technology for enhancing the EE of the ISAC system, the multiplicative fading feature hinders its effectiveness. This paper proposes the use of active RIS with its amplification gains to assist the ISAC system for EE improvement. Specifically, we formulate an EE optimization problem in an active RIS-aided ISAC system under system power budgets, considering constraints on user communication quality of service and sensing signal-to-noise ratio (SNR). A novel alternating optimization algorithm is developed to address the highly non-convex problem by employing the generalized Rayleigh quotient optimization, semidefinite relaxation (SDR), and the majorization-minimization (MM) framework. Furthermore, to reduce computational complexity, we derive a semi-closed form for eigenvalue determination. Numerical results demonstrate the effectiveness of the proposed approach, showcasing significant improvements in EE compared to both passive RIS and spectrum efficiency optimization cases.
\end{abstract}

\begin{IEEEkeywords}
Integrated sensing and communication, Active RIS, Energy efficiency, Generalized Rayleigh quotient optimization, Semi-definite relaxing, Majorization-minimization
\end{IEEEkeywords}

\section{Introduction}
\IEEEPARstart{T}{he}  widespread adoption of fifth-generation (5G) wireless communication has brought significant advancements, including massive multi-input multi-output (MIMO), millimeter-wave (mmWave) communication, and ultra-dense networks, enabling high data transmission rates, low latency, and massive device access \cite{intro_5G}. However, these advancements have also introduced new challenges, especially the serious spectrum scarcity and high energy consumption \cite{6G_challenge}. With the exponential growth of connected devices and the increasing demand for high-bandwidth applications, the available spectrum is rapidly becoming insufficient. Additionally, the pursuit of high spectrum efficiency (SE) often comes at the expense of high energy consumption, which leads to a low energy efficiency (EE) and is not aligned with the vision of green communication. Therefore, addressing these challenges is crucial for the sustainable development of 6G networks.

To address the growing spectrum scarcity,  integrated sensing and communication (ISAC) has emerged, enabling communication systems to share the spectrum with radar systems \cite{ISAC_concept}. ISAC can be categorized into three main types, namely, coexistence, collaboration, and joint design \cite{ISAC_3_theme}. In coexistence ISAC, the communication and radar systems share the spectrum, with interference being a primary concern, which can be solved by opportunistic spectrum access \cite{Opportunistic_access}, null space projection \cite{NSP}, and transceiver design \cite{ISAC_interf_Rihan}, etc. Collaboration ISAC involves sharing both spectrum and essential knowledge, such as target direction and communication symbols, between the two systems. The shared information enables more thorough interference cancellation. In \cite{Com_Sensing_collaborate}, a typical collaboration system was studied where communication signals assisted environment sensing, and sensing results aided symbol detection. Joint design ISAC integrates the systems comprehensively by sharing spectrum, hardware, waveforms, and other components simultaneously, where the key challenge lied in a common waveform designs and circuit sharing. In \cite{ISAC_LF}, the authors designed probing beampatterns while guaranteeing downlink communication performance. In \cite{Joint_Transmit_Beamforming}, a joint transmit beamforming model was proposed by designing the weighted sum of independent radar waveforms and communication symbols. {Furthermore, in \cite{Zhengchuanaffili}, rate-splitting multiple access was employed in ISAC systems to achieve effective multiple access, and radar beampattern design.}  The work in \cite{El3} reduced the consumption of spectrum and hardware resources with three methods for ISAC waveforms design. In \cite{radar_ISAC_LGS}, a joint range and velocity estimation method was developed considering an orthogonal frequency division multiplexing (OFDM) ISAC waveform. In addition, the authors in \cite{El2} designed the waveform with offset quadrature amplitude modulation based OFDM in a ISAC system. {Albeit the above works of joint waveform design of ISAC were effective, EE was not adequately considered. The work in \cite{ISACEE} tried to fill the gap and developed an energy-efficient interference cancellation approach in ISAC systems.}

In parallel with the development of ISAC, reconfigurable intelligent surface (RIS) has also been a candidate in future communication systems. RIS consists of numerous reflecting elements and modifies the propagation channels by intelligently altering the phase of signals \cite{IRS_cz}. This capability has led to a surge of research of RISs in various scenarios. In \cite{RIS_RR_hexin}, a relax-and-retract method was proposed to minimize transmit power under quality-of-service (QoS) constraints. In \cite{IRS_mmWave_fangjun}, the authors exploited the structure of mmWave channels to derive a closed-form solution for single RIS scenarios and a near-optimal analytical solution for multi-RIS scenarios. {Apart from unicast transmission, a RIS-assisted multigroup multicast communication system was investigated in \cite{Zhengchuanthree} and ergodic rate was optimized to effectively improve SE.} {Moreover, the EE maximization of a RIS-assisted multi-user communication system was investigated in \cite{EE_RIS}.} Beyond communication, RISs have also shown great promises in other areas, such as RIS-aided radar detection \cite{RIS_radar}, RIS-aided simultaneous wireless information and power transfer \cite{RIS_SWIPT}, RIS-aided localization \cite{El4} and RIS-aided physical layer security  \cite{RIS_secure}. {
Additionally, comprehensive overviews on signal processing for RIS-aided wireless systems, including channel estimation, data transmission, and localization schemes, can be found in \cite{CunhuaRISSignalProcess}, \cite{SignalProcess}.}

Despite the above numerous benefits, the multiplicative fading effects posed by traditional passive RISs limit the system performance, where the equivalent path loss of the transmitter-RIS-receiver link is the product of the path loss of the transmitter-RIS and RIS-receiver links. To address this issue, active RISs have been proposed, which can amplify the incident signal and modify its phase shift to provide greater performance gains \cite{ARIS_you,ARIS_ori2}. Some studies have revealed the merits of active RISs. In \cite{ARIS_ori1}, an active RIS model was developed, and the RIS beamformer was optimized to maximize the receiving power while minimizing the RIS-related noise. The authors in \cite{Active_RIS_CSI}  optimized average achievable rate considering partial CSI of the RIS-aided channels. {Unlike those works of cell-centric networks, the authors in \cite{Zhengchuanone} considered an active RIS-aided user-centric networks for inter-cell interference reduction.} Beyond SE optimizations, EE maximization was also a attractive topic in RIS systems. {In \cite{ActPasRISCunhua}, the EE of active RIS was compared to that of passive RIS under the same power budget in communication systems.}  In \cite{ARIS_MultiAccess}, active RISs were employed to assist uplink access and transmission from multiple devices in an energy-constrained system.  In \cite{ARIS_sub}, a sub-connected architecture of active RISs was proposed to reduce power consumption and improve EE. In \cite{Rihan_2024}, communication energy efficiency aided by active RIS was maximized under a constraint of power budget for a probing process. These works validated the significance of active RISs for enhancing SE, EE, and other performances in various scenarios. 

Taking full advantages of both technologies, the synergy of ISAC and RISs shows great promises for enhancing SE, EE, and various communication performances. Extensive research has been conducted on RIS-assisted ISAC systems. In \cite{RIS_ISAC_jiang}, the authors jointly optimized transmit beamforming and RIS phase shifts to maximize radar signal-to-noise ratio (SNR) while maintaining communication SNR above a threshold. Similarly, in \cite{RIS_ISAC}, the radar output signal-to-interference-plus-noise ratio (SINR) was maximized using space-time adaptive processing under various communication QoS metrics. In \cite{RIS_ISAC_PD}, the user SNR was maximized under a minimum radar detection probability constraint. The authors in \cite{El1} considered a fair sensing-communication waveform design with RIS by simultaneously maximizing the sensing SINR and minimize the multi-user interference. {Additionally, a rate splitting multiple access-based RIS-assisted ISAC system was proposed in \cite{Zhengchuantwo} to simultaneously achieve target detection and multi-user non-orthogonal multiple access communication.} {In \cite{PassiveRISDFRCTGCN}, an alternative approach was proposed for EE optimization in a passive RIS-aided ISAC system.} However, the multiplicative fading effects of passive RISs still limit the ISAC system performance gains, where active RISs can be introduced to tackled the issue. A few works have studied the potentials of active RIS-assisted ISAC systems. In \cite{JJ}, the authors maximized the sum rate of users in an active RIS-aided ISAC scenario, subject to power budget and detection power constraints. In \cite{ARIS_ISAC_SNR}, an active RIS-assisted ISAC system was designed to maximize radar output SNR while satisfying communication SINR requirements. In \cite{ARIS_ISAC_SNR2}, beamforming design and performance analysis was studied in active RIS-assisted ISAC system. In \cite{ARIS_ISAC_secure}, a secure active RIS-aided ISAC system was proposed, maximizing secrecy rate under minimum radar detection SNR and total power budget constraints. {Although previous studies on the active RIS-aided ISAC system have prioritized achieving exceptional performance in specific areas, they have overlooked the efficiency of energy utilization. Given that ISAC base stations (BS) and active RISs typically consume significant amounts of energy, addressing the relationship between power consumption and the enhancement of communication and sensing performance remains a critical challenge.} 

{In contrast to the outlined background, our research delves into the EE performance within an active RIS-assisted ISAC system. We introduce a novel framework centered around formulating an EE maximization problem and tailoring an alternating optimization approach to efficiently address this challenge. Key contributions of this study encompass:}

\begin{itemize}
\item {In pursuit of sustainable communication practices, this study examines and tailors a comprehensive optimization framework to maximize system EE. This framework encompasses joint design considerations for the transceivers at the BS, amplification, and phase shift matrix of the active RIS. Striking a balance between maintaining sensing and communication performance, the framework imposes constraints on the SNR of the target echo and the SINR of users. Furthermore, it ensures that power consumption from both the transmitter and active RIS remains within given thresholds. }
\item {To effectively tackle the multifaceted optimization problem,  an alternating optimization approach is customized, which decouples the fractional objective function with Dinkelbach algorithm and optimizes each variable alternatively. Additionally, the optimization framework flexibly employes generalized Rayleigh quotient optimization, semidefinite relaxation (SDR) and majorization-minimization (MM) techniques to design the transmit and receiving beamformer, and the amplification and phase shift matrix of the active RIS, respectively. }
\item {To facilitate the dimension reductions of the objective function in MM framework, a surrogate function is selected by proving an inequality associated with eigenvalues. Additionally, with the aim of reducing the high computational complexity arising from eigenvalue decomposition when constructing the surrogate functions, semi-closed forms of the eigenvalues are derived by exploiting the structural characteristics of the matrix.}
\item {Numerical results validate the effectiveness and excellent convergences of the proposed approach, revealing that the integration of active RIS in the ISAC systems improves EE by significantly reducing the power consumption and keeping SE. In addition, the proposed framework achieves remarkable EE enhancements compared to passive RIS scenarios and spectrum optimization cases. }
\end{itemize}

This paper is structured as follows. Section \ref{system model} introduces the system model and formulates the EE maximization problem. Section \ref{algorithm} presents the proposed alternating optimization algorithm for joint design of transceivers and active RIS  beamformer. Section \ref{simulation} provides comprehensive simulation results to demonstrate the effectiveness of the proposed approach. Finally, conclusions are drawn in Section \ref{conclusions}.

\textit{Notation}: Boldface lowercase and uppercase letters represent vectors and boldface, respectively. The operators $(\cdot)^H$, $(\cdot)^T$, and $(\cdot)^*$ denote conjugate transpose, transpose, and conjugate operations, respectively. The symbol $\mathrm{tr}(\cdot)$ is the trace of matrix. The symbol $\mathrm{diag}(\mathbf{A})$ denotes a vector whose entries are the diagonal elements of matrix $\mathbf{A}$, while $\mathrm{diag}(\mathbf{a})$ is a diagonal matrix whose diagonal elements are the elements in vector $\mathbf{a}$. The symbols $\otimes$ and $\odot$ are Kronecker and Hadamard product, respectively. The operator $\mathrm{vec}(\mathbf{A})$ stands for vectorization of matrix $\mathbf{A}$, while $\mathrm{unvec}(\mathbf{a})$ is an inverse operation of $\mathrm{vec}(\mathbf{A})$. The symbols $|\cdot|$ and $\|\cdot\|$ denote absolute value and norm operations, respectively. $\mathbb{C}^{M \times N}$ is the complex space of $M \times N$ dimensions. $\mathcal{CN}\left( 0, \sigma^2 \right)$ indicates that a random variable follows a complex Gaussian distribution with zero mean and variance equals  to $\sigma^2$. $\mathfrak{R}[x]$ and $\angle x$ represent the real part and the phase of the complex symbol $x$, respectively. The operation $\mathbb{E}[\cdot]$ takes the mean of a random variable.

\section{System model And Problem Formulation} \label{system model}
\subsection{System Model}
\begin{figure}[!t]
    \centering
    \includegraphics[scale=0.36]{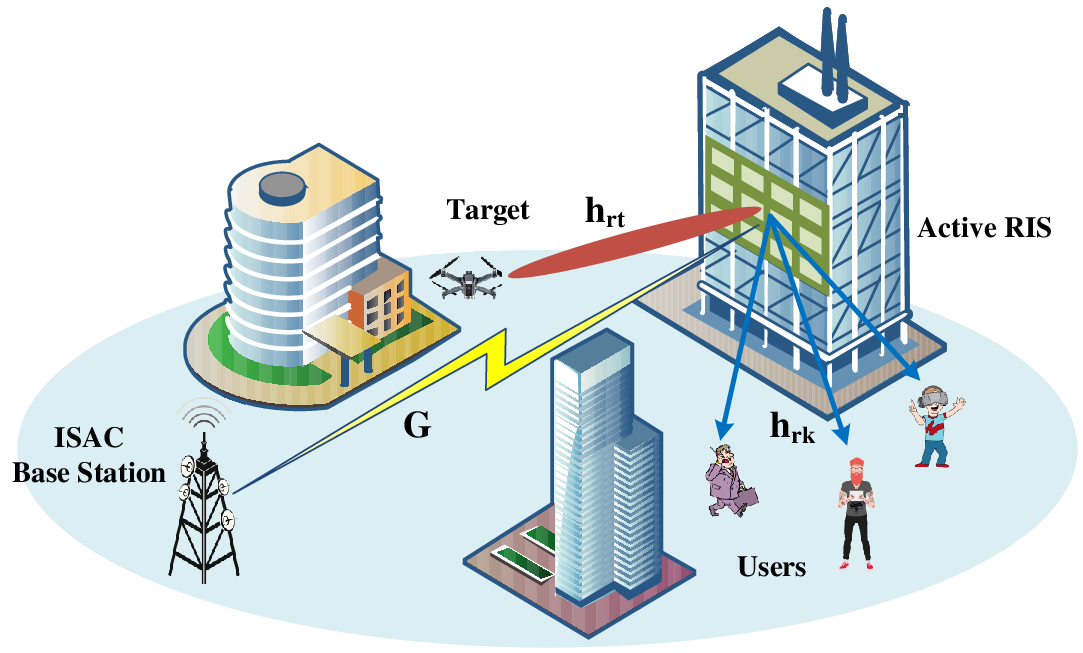}
    \caption{A scenario of an active RIS-aided target sensing and multi-user communication.}
    \label{fig:System model}
\end{figure}

Consider an active RIS-aided ISAC system as depicted in Fig. \ref{fig:System model}. The system comprises an $M$-antenna ISAC BS that serves $K$ single-antenna users and simultaneously detects a single target. {The BS is mono-static so that it is able to transmit integrated communication and sensing signals and receive the target's echo.} Without loss of generalization, the paths from the BS to the users and the target both are considered being  blocked{\footnote{{The proposed method remains applicable even when a direct link exists. The presence of a direct link does not affect the transceiver design, as they are based on a fixed compound channel. Additionally, when designing the active RIS's scattering matrix, the direct link introduces extra linear and quadratic terms in the objective function and constraints. These can be incorporated into the existing formulations. Therefore, the ultimate form of the optimization problem remains unchanged and can be solved using the proposed method.}}}. To establish virtual links for both communication and sensing, an active RIS, composed of $N$ reflecting elements, is deployed. Additionally, all the channels  are assumed to be known{\footnote{{The estimation of communication and sensing channels is beyond the scope of this work. In practice, before the actual transmission of communication data, users send pilots, and a two-phase strategy can be utilized for low-overhead communication channel estimation \cite{CEE}. Meanwhile, the structured sensing channel can be obtained by estimating the direction-of-arrival and radar cross-section of the target \cite{RISDOA}.}}}. 

\subsubsection{Transmit Signal}
The signal transmitted by the BS is a superposition of communication symbols $\mathbf{s_c} \in \mathbb{C}^{K}$ and radar waveform $\mathbf{s_t} \in \mathbb{C}^{M}$. To accommodate both communication and sensing functionalities within the transmit signal, the BS employs beamforming to focus the signal towards the intended users as well as the target. Mathematically, the transmit signal can be represented as
\begin{align}
    \mathbf{W s} &= \mathbf{W_t s_t} + \mathbf{W_c s_c} = \begin{bmatrix}
        \mathbf{W}_c & \mathbf{W}_t
    \end{bmatrix}
    \begin{bmatrix}
        \mathbf{s}_c \\ \mathbf{s}_t
    \end{bmatrix},
\end{align}
where $\mathbf{W} \in \mathbb{C}^{M \times (K+M)}$ denotes the joint beamforming matrix, encapsulating both communication beamformer $\mathbf{W_c} \in \mathbb{C}^{M \times K}$ and sensing beamformer $\mathbf{W_t} \in \mathbb{C}^{M \times M}$ as $\mathbf{W} =  \begin{bmatrix}
        \mathbf{W}_c & \mathbf{W}_t
    \end{bmatrix} \in \mathbb{C}^{M \times (M+K)}$. The integrated baseband waveform, denoted by $\mathbf{s} \in \mathbb{C}^{M+K}$, represents the combined communication and sensing signals transmitted by the BS. It is constructed by stacking the communication symbols $\mathbf{s}_c \in \mathbb{C}^K$ and the radar waveform $\mathbf{s}_t \in \mathbb{C}^M$ into a single vector {\footnote{{
The mixed signals are not separated at the active RIS because it lacks the capability to process baseband signals \cite{ARIS_ori1}. At the user receivers, the transmit beamformers and RIS beamformers are designed to suppress the sensing signal and multi-user interference, allowing for the detection of the desired symbols \cite{Joint_Transmit_Beamforming}. In the mono-static ISAC receiver, the communication symbols can either be subtracted from the mixed signal \cite{mixone} or regarded as part of the sensing waveform \cite{Joint_Transmit_Beamforming}, \cite{mixtwo}, since they are already known to the sensing receiver.}}}, $
\mathbf{s} = \begin{bmatrix}
    \mathbf{s}_c^T & \mathbf{s}_t^T
  \end{bmatrix}^T\in \mathbb{C}^{M+K}
$. {Without loss of generality, the communication symbols $\mathbf{s}_c$ and the sensing signal $\mathbf{s}_t$ ought to satisfy the following requirements \cite{Joint_Transmit_Beamforming}, \cite{Radarwaveformthree}. Both
$\mathbf{s}_c$ and $\mathbf{s}_t$ are zero-mean stationary stochastic process, i.e, $ \mathbb{E}[\mathbf{s}_c]=\mathbb{E}[\mathbf{s}_t]=\mathbf{0}^{M}$, while they are statistically independent and uncorrelated, namely $ \mathbb{E}[\mathbf{s}_t \mathbf{s}_c^H]=\mathbf{0}^{M \times K}$. In addition, the communication symbols in $\mathbf{s}_c$ are independent and normalized to unit power, mathematically expressed as $ \mathbb{E}[\mathbf{s}_c\mathbf{s}^H_c]=\mathbf{I}^{K \times M}$. Similarly, $\mathbf{s}_t$ has $M$ unit-power and spatially orthogonal individual waveforms, namely $ \mathbb{E}[\mathbf{s}_t \mathbf{s}^H_t]=\mathbf{I}^{M \times M}$. Typically, the communication symbols are the digital modulated symbols such as PSK and QAM \cite{Radarwaveformone}, while the sensing signals can be pseudo
random coding sequence \cite{Joint_Transmit_Beamforming}, \cite{Radarwaveformtwo}. } 

\subsubsection{Communication Model}
As illustrated in Fig.~\ref{fig:System model}, the BS transmits the integrated signals, which are then reflected by the active RIS and finally received by the users. We denote the channel from the BS to the active RIS as $\mathbf{G} \in \mathbb{C}^{N \times M}$, and the channel from the active RIS to user $k$ as $\mathbf{h}_{r,k}^H \in \mathbb{C}^{1 \times N}$. The active RIS beamformer, including both amplification and phase shift, are represented by
$\mathbf{A} = \mathrm{diag}([a_1, \dotsc, a_N])$ and $ \mathbf{\Theta} =\mathrm{diag}([\theta_1, \dotsc, \theta_N])$, where $a_n$ and $\theta_n$ denote the amplification and phase shift coefficients of the $n$-th element, respectively. For simplicity, we define $\mathbf{\Psi} = \mathbf{A} \mathbf{\Theta} = \mathrm{diag}([\psi_1, \dotsc, \psi_N])$. It is important to note that in practical scenarios, the amplifiers operate within the linear amplification zone, restricting the amplification coefficient of each element $a_n$ to a maximum value $a_{max}$, i.e., $a_n \leq a_{max}$. Consequently, the signal received by user $k$ can be expressed as
\begin{align} 
\label{comm_model}
y_k &= \mathbf{h}_{r,k}^H \mathbf{\Psi}^H \mathbf{G} \mathbf{W s} + \mathbf{h}_{r,k}^H \mathbf{\Psi}^H \mathbf{z} + n_k.
\end{align}
The dynamic noise introduced by the amplifier of active RIS is represented by the vector $\mathbf{z} \in \mathbb{C}^N$, while $n_k$ denotes the noise of the $k$-th user receiver. The noise sources are modeled as additive white Gaussian noise (AWGN) with zero mean. The noise variances of the users' receivers and active RIS are denoted by $\sigma_k^2$ and $\sigma_0^2$, respectively.

Following the received signal model presented in (\ref{comm_model}), the SINR of user $k$ can be expressed as
\begin{align}
    \gamma_k = \frac{|\mathbf{h}_k^H \mathbf{w}_k|^2}{\sum_{i \neq k}^{K+M}|\mathbf{h}_k^H \mathbf{w}_i|^2 + \|\mathbf{h}_{r,k}^H \mathbf{\Psi}^H\|^2 \sigma_0^2 + \sigma_k^2}, 
\end{align}
where $\mathbf{h}_k^H=\mathbf{h}_{r,k}^H \mathbf{\Psi}^H \mathbf{G}$ represents the cascaded communication channel for user $k$. Subsequently, the sum-rate can be determined as follows:
\begin{align}
    R_b = \sum_{k=1}^{K}\log_2 (1+\gamma_k).
\end{align}

\subsubsection{Radar Model}
{
For target sensing, the transmit signals first propagate through the active RIS and illuminate the target. The signal at the target can be expressed similarly to (\ref{comm_model}) as 
\begin{align}
    \mathbf{e}_1 = \mathbf{h}_{rt}^H \mathbf{\Psi}^H \mathbf{G} \mathbf{W s} + \mathbf{h}_{rt}^H \mathbf{\Psi}^H \mathbf{z},
\end{align}
 where $\mathbf{h}^H_{rt} \in \mathbb{C}^{1\times N}$ represents the channel from the active RIS to the target. The echo is then reflected from the target and processed by the active RIS again, resulting in
 \begin{align}
    \mathbf{e}_2 = \mathbf{\Psi}\mathbf{h}_{rt}\mathbf{e}_1+ \mathbf{\Psi}\mathbf{z}_1.
 \end{align}
The vector $\mathbf{z}_1 \!\!\!\in\!\! \!\mathbb{C}^{N}$ denotes the AWGN at the output of the active RIS during backward propagation, following a distribution of $\mathcal{CN}(0, \sigma_0^2)$. Subsequently, the signal is received by the BS after passing through the channel $\!\mathbf{G}^H\!$, expressed as{\footnote{{The target can not only be the drones in the sky, but also the vehicles on the ground. For the target in the sky, the model and the algorithm are exactly shown in the paper. For the target on the ground, the sensing model should also consider the existence clutter. In this case, the algorithm can be easily extended, where the newly-added term can be integrated with existing terms. Similar extension can be found in \cite{Skyone}, \cite{skytwo}.}} }
\begin{align} \label{radar model}
    \mathbf{y}_t & = \mathbf{G}^H\mathbf{e}_2 + \mathbf{n}_r  =\mathbf{H}_t \mathbf{Ws}+ \mathbf{H}_{z0} \mathbf{z}+\mathbf{H}_{z1} \mathbf{z}_1 +\mathbf{n}_r,
\end{align}
where 
\begin{subequations}
    \begin{align}
    &\mathbf{H}_t = \mathbf{G}^H \mathbf{\Psi}\mathbf{h}_{rt}\mathbf{h}_{rt}^H \mathbf{\Psi}^H \mathbf{G}, \\
    &\mathbf{H}_{z0} = \mathbf{G}^H \mathbf{\Psi}\mathbf{h}_{rt} \mathbf{h}_{rt}^H \mathbf{\Psi}^H , \\
    &\mathbf{H}_{z1} = \mathbf{G}^H \mathbf{\Psi}.
\end{align}
\end{subequations}
In the receiving model (\ref{radar model}), the term $\mathbf{H}_t \mathbf{Ws}$ represents the echo from the target, where $\mathbf{H}_t \in \mathbb{C}^{M \times M}$ is the channel matrix between the target and the receiver. The term $\mathbf{H}_{z0} \mathbf{z}$ accounts for the amplifier-induced dynamic noise during forward propagation through the active RIS, while the term $\mathbf{H}_{z1} \mathbf{z}_1$ accounts for the amplifier-induced dynamic noise during backward propagation through the active RIS. The vector $\mathbf{n}_r$ represents the noise of the echo receiver, which is assumed to be AWGN with zero mean and variance $\sigma_r^2$.}

To enhance the echo SNR, a receiving beamformer $\mathbf{u}^H \in \mathbb{C}^{1 \times M}$ is applied to the echo, which is given by
\begin{align}
    \mathbf{u}^H \mathbf{y}_t &= \mathbf{u}^H \mathbf{H}_t \mathbf{Ws} + \mathbf{u}^H \mathbf{H}_{z0} \mathbf{z}+\mathbf{u}^H \mathbf{H}_{z1} \mathbf{z}_1+\mathbf{u}^H\mathbf{n}_r. 
\end{align}
As a result, the echo SNR for target sensing is given as
\begin{align}
    \gamma_r = \frac{\mathbf{u}^H \mathbf{H}_t \mathbf{W W}^H \mathbf{H}_t^H \mathbf{u}}{\mathbf{u}^H(\sigma_0^2 \mathbf{H}_{z0} \mathbf{H}_{z0}^H+\sigma_0^2 \mathbf{H}_{z1} \mathbf{H}_{z1}^H+ \sigma_r^2 \mathbf{I}_M) \mathbf{u}}.
\end{align}
Since the echo SNR is directly related to various radar performance metrics, such as detection probability and parameter estimation accuracy, it can be used as a key performance indicator (KPI) for target sensing \cite{ISAC_concept}. {It is worth mentioning that we adopt a widely accepted assumption that the sensing channel can be acquired through previous estimation stages \cite{Radarwaveformone,whysensix} or from an information map database \cite{RIS_ISAC,whysenthree}. In this context, sensing primarily serves to fine-tune coarse estimated parameters \cite{whysensix}, track and update target conditions \cite{whysentwo}, and establish the groundwork for higher-level sensing functions \cite{whysenfour}.} 

\subsubsection{Power Consumption}
The system power consumption consists of two primary components: the power consumed at the BS and the power consumed at the active RIS. At the BS, the power consumption encompasses the transmit power $\|\mathbf{W}\|_F^2$ and the static hardware power $P_{ST}$, as represented by
\begin{equation}
P_{ISAC} = \xi \|\mathbf{W}\|_F^2 + P_{ST},
\end{equation}
where the symbol $\xi$ denotes the inverse of the energy conversion coefficient at the ISAC transmitter. The power consumption at the active RIS, on the other hand, comprises the power utilized for signal amplification and the hardware power required for controlling the amplifier and phase tuner at the active RIS. Therefore, the power consumed at the active RIS can be expressed as
\begin{align}
    P_{RIS}&= N \left(P_{RP} + P_{RA}\right)+ \zeta \left(\|\mathbf{\Psi}^H \mathbf{G}\mathbf{W}\|_F^2 + 2\sigma_0^2\|\mathbf{\Psi}\|_F^2\right.  \notag \\
    &+ \|\mathbf{\Psi} \mathbf{h}_{rt} \mathbf{h}_{rt}^H \mathbf{\Psi}^H \mathbf{G} \mathbf{W} \|_F^2+ \sigma_0^2 \| \mathbf{\Psi}\mathbf{h}_{rt} \mathbf{h}_{rt}^H \mathbf{\Psi}^H \|_F^2 ),
\end{align}
where $\zeta$ represents the inverse of the energy conversion coefficient at the active RIS, while $P_{RP}$ and $P_{RA}$ correspond to the hardware power of the phase tuner and the power amplifier at the active RIS, respectively. The total power consumption of the system is given as
\begin{equation}
P = P_{ISAC}+ P_{RIS}.
\end{equation}
Following the definition in \cite{PassiveRISDFRCTGCN}, the system EE is given as
\begin{equation} \label{EE_concept}
\eta = \frac{R_b}{P},
\end{equation}
which represents the amount of transmitted data per Joule.

\subsection{Problem Formulation}
In this paper, we aim to optimize the system EE by jointly designing the transmit beamforming matrix $\mathbf{W}$, the amplification matrix $\mathbf{A}$, the phase shift matrix $\mathbf{\Theta}$ of the active RIS, and the echo receiving beamformer $\mathbf{u}$. This optimization problem can be formulated as
\begin{subequations}\label{formulation}
\begin{align} 
\max _{\mathbf{W},\mathbf{\Theta},\mathbf{A},\mathbf{u}} &~  \eta = \frac{R_b}{P}, \tag{\ref{formulation}{a}} \\ 
s.t.&~  P_{ISAC} \leq P_{ISAC}^{max}, \tag{\ref{formulation}{b}} \\
& P_{RIS} \leq P_{RIS}^{max},  \tag{\ref{formulation}{c}} \\
& \gamma_k \geq \tau_k, \forall k = 1,\dots,K , \tag{\ref{formulation}{d}} \\
& \gamma_r \geq \tau_r,  \tag{\ref{formulation}{e}}\\
& a_n \leq a_{max}, \forall n=1,\dots,N, \tag{\ref{formulation}{f}}
\end{align}
\end{subequations}
where constraints (\ref{formulation}{b}) and (\ref{formulation}{c}) represent the power budgets of the BS and the active RIS, respectively, with $P_{ISAC}^{max}$ and $P_{RIS}^{max}$ denoting the maximum permissible power consumption at the BS and the active RIS. Constraint (\ref{formulation}{d}) ensures the communication QoS for each user, requiring that the SINR of each user remains above the predefined SINR threshold $\tau_k$. Constraint (\ref{formulation}{e}) stipulates that the echo SNR must exceed the threshold $\tau_r$ to guarantee satisfactory sensing performance. Finally, constraint (\ref{formulation}{f}) imposes a maximum amplification limit of $a_{max}$ for each active RIS element. It is challenging to solve problem (\ref{formulation}) due to the non-convexity of the objective function and the constraints (\ref{formulation}{b}), (\ref{formulation}{c}), (\ref{formulation}{d}), and (\ref{formulation}{e}). 

\section{Proposed Alternative Optimization Algorithm} \label{algorithm}

In this section, an alternating optimization-based approach is proposed to tackle the challenging problem (\ref{formulation}). Initially, a pre-processing step is applied to transform the objective function into a more tractable form. Subsequently, the problem is decomposed into three subproblems and optimized iteratively based on the generalized Rayleigh quotient optimization, SDR, and MM frameworks. 

\subsection{Pre-processing} \label{pre-processing}
Prior to solving the optimization process, a pre-processing step should be applied to the objective function due to its intractable fractional form. To decouple the numerator and denominator, we employ the classic Dinkelbach algorithm \cite{Dinkelbach}. Specifically, by introducing an auxiliary variable $\mathbf{\mu}$, the original problem (\ref{formulation}) can be equivalently reformulated as
\begin{subequations}\label{formulation2}
\begin{align} 
\max _{\mu,\mathbf{W},\mathbf{\Psi},\mathbf{u}} &~ R_b - \mu P, \tag{\ref{formulation2}{a}} \\ 
s.t.&~  (\ref{formulation}{b}), (\ref{formulation}{c}), (\ref{formulation}{d}), (\ref{formulation}{e}), (\ref{formulation}{f}) , \tag{\ref{formulation2}{b}}
\end{align}
\end{subequations}
where the optimal value of $\mu$ is $\mu^\star=\frac{R_b}{P}$. However, when $\mu$ is fixed, $R_b$ remains in a complicated form due to the logarithm and the fractional form. To separate the SNR from the logarithm form and decouple the fractional form of SNR, we employ the Lagrangian dual transform \cite{WJZ1} \cite{WJZ2} and quadratic transform techniques \cite{Lagrange_QT} to transform the objective function into a more tractable form, which is given by
\begin{align}
f(\mathbf{W},\mathbf{\Psi},\pmb{\rho},\pmb{\nu}) =g(\mathbf{W},\mathbf{\Psi},\pmb{\rho},\pmb{\nu}) -\mu(P_{ISAC}+ P_{RIS}),
\end{align}
where
\begin{align}    &\!\!\!g(\mathbf{W},\mathbf{\Psi},\pmb{\rho},\pmb{\nu})\!=\!\sum_{k=1}^{K} \!\left[\ln (1+\rho_k)\!-\!\rho_k
\!+\! 2\sqrt{1\!+\!\rho_k} \mathfrak{R}\left\lbrace \nu_k^*  \mathbf{h}_k^H \mathbf{w}_k \right\rbrace  \right. \notag\\
   & \left. - |\nu_k|^2 \left( \sum_{i=1}^{K+M}|\mathbf{h}_k^H \mathbf{w}_i|^2 + \|\mathbf{h}_{r,k}^H \mathbf{\Psi}^H\|^2 \sigma_0^2 + \sigma_k^2 \right) \right]
\end{align}
is the equivalent form of $R_b$ after the transformation of Lagrangian dual transform and quadratic transform. Specifically, the auxiliary variables $\pmb{\rho}$ and $\pmb{\nu}$ are introduced to decouple the logarithm and the fractional form, respectively. Given $\mathbf{W}$ and $\mathbf{\Psi}$, the problem becomes an unconstrained optimization problem with respect to $\pmb{\rho}$ and $\pmb{\nu}$. Consequently, the optimal solutions $\pmb{\rho}^\star$ and $\pmb{\nu}^\star$ can be obtained by solving $\frac{\partial f(\pmb{\rho})}{\partial\pmb{\rho}}=0$ and $\frac{\partial f(\pmb{\nu})}{\partial\pmb{\nu}}=0$, given as
\begin{align} 
    \rho^\star_k &= \frac{\iota_k}{2} (\iota_k + \sqrt{\iota_k^2+4}), \label{auxi1} \\
    \nu^\star_k &= \frac{\sqrt{1+\rho_k}\mathbf{h}_k^H \mathbf{w}_k}{\sum_{i=1}^{K+M}|\mathbf{h}_k^H \mathbf{w}_i|^2 + \|\mathbf{h}_{r,k}^H \mathbf{\Psi}^H\|^2 \sigma_0^2 + \sigma_k^2}, \label{auxi2} 
\end{align}
where we denote $\iota_k= \mathfrak{R}\left\lbrace \nu_k^*  \mathbf{h}_k^H \mathbf{w}_k \right\rbrace$. With optimized auxiliary variables, the problem (\ref{formulation2}) can then be simplified by omitting the constants, expressed as
\begin{subequations}\label{formulation3}
\begin{align} 
\max _{\mathbf{W},\mathbf{\Psi},\mathbf{u}} &~ f_1(\mathbf{W},\mathbf{\Psi}), \tag{\ref{formulation3}{a}} \\ 
s.t.&~  (\ref{formulation}{b}), (\ref{formulation}{c}), (\ref{formulation}{d}), (\ref{formulation}{e}), (\ref{formulation}{f}) ,\tag{\ref{formulation3}{b}}
\end{align}
\end{subequations}
where 
\begin{align} 
&f_1(\mathbf{W},\mathbf{\Psi})=\sum_{k=1}^{K} \left[ 2\sqrt{1+\rho_k} \mathfrak{R}\left\lbrace \nu_k^*  \mathbf{h}_k^H \mathbf{w}_k\right\rbrace- |\nu_k|^2 \right. \notag\\
   & \left. \left( \!\sum_{i=1}^{K+M}\!\!|\mathbf{h}_k^H \mathbf{w}_i|^2 \!+\! \|\mathbf{h}_{r,k}^H \mathbf{\Psi}^H\|^2 \sigma_0^2 \right) \!\right]\!\!-\!\!\mu(P_{ISAC}\!\!+\! \!P_{RIS}).
\end{align}

\subsection{Echo Beamforming  Design} \label{echo beamfomer}
Considering the design of receiving beamforming $\mathbf{u}$, it is evident that, in the problem (\ref{formulation3}), only the constraint (\ref{formulation}{e}) involves $\mathbf{u}$, which is related to the echo SNR. Therefore, the design of $\mathbf{u}$ can be directly transformed to the maximization of the target echo SNR $ \gamma_r$, expressed as
\begin{align} \label{opt_u}
\max _{\mathbf{u}} &~  \frac{\mathbf{u}^H \mathbf{C} \mathbf{u}}{\mathbf{u}^H\mathbf{D} \mathbf{u}},
\end{align}
where we denote $\mathbf{C}\!=\!\mathbf{H}_t \mathbf{W W}^H \mathbf{H}_t^H \!\in\! \mathbb{C}^{M \times M} $ and $\mathbf{D}\!=\!(\sigma_0^2 \mathbf{H}_{z0} \mathbf{H}_{z0}^H\!+\!\sigma_0^2 \mathbf{H}_{z1} \mathbf{H}_{z1}^H\!+ \!\sigma_r^2 \mathbf{I}_M ) \!\in\!\mathbb{C}^{M \times M}$. It is obvious that $\mathbf{C}$ and $\mathbf{D}$ are Hermitian matrices, with $\mathbf{D}$ being a positive semidefinite matrix. Consequently, the problem (\ref{opt_u}) falls under the category of generalized Rayleigh quotient optimization problems, which is a well-established class of optimization problems {\cite{Matrix}}. The optimal solution to this problem is the eigenvector of $\mathbf{D}^{-1} \mathbf{C}$ associated with the largest eigenvalue.

\subsection{Transmit Beamforming  Design} \label{transmit beamformer}
In this subsection, we address the optimization of transmit beamforming $\mathbf{W}$. With the remaining variables held constant, the optimization problem (\ref{formulation3}) can be simplified as 
\begin{subequations}\label{opt_w}
\begin{align} 
\max _{\mathbf{W}} &~  \sum_{k=1}^{K} \!\left[
2\sqrt{1\!+\!\rho_k} \mathfrak{R}\left\lbrace \nu_k^*  \mathbf{h}_k^H \mathbf{w}_k \right\rbrace\! - \!|\nu_k|^2\!\sum_{i=1}^{K+M}\!|\mathbf{h}_k^H \mathbf{w}_i|^2  \right] \notag\\
& -\mu \left[\xi \|\mathbf{W}\|_F^2+ \zeta \left(\|\mathbf{H}_{z1}^H \mathbf{W}\|_F^2  + \|\mathbf{H}_{z0}^H \mathbf{W} \|_F^2 \right)\right],\tag{\ref{opt_w}{a}} \\ 
s.t.&~  \|\mathbf{W}\|_F^2 \leq \widetilde{P}_{ISAC}, \tag{\ref{opt_w}{b}} \\
\vspace{10pt}
&~\|\mathbf{H}_{z1}^H \mathbf{W}\|_F^2  + \|\mathbf{H}_{z0}^H \mathbf{W} \|_F^2  
    \leq \widetilde{P}_{RIS},   \tag{\ref{opt_w}{c}} \\
& \frac{|\mathbf{h}_k^H \mathbf{w}_k|^2}{\sum_{i \neq k}^{K+M}|\mathbf{h}_k^H \mathbf{w}_i|^2 + \|\mathbf{h}_{r,k}^H \mathbf{\Psi}^H\|^2 \sigma_0^2 + \sigma_k^2} \geq \tau_k,   \tag{\ref{opt_w}{d}} \\
& \frac{\mathbf{u}^H \mathbf{H}_t \mathbf{W W}^H \mathbf{H}_t^H \mathbf{u}}{\mathbf{u}^H(\sigma_0^2 \mathbf{H}_{z0} \mathbf{H}_{z0}^H+\sigma_0^2 \mathbf{H}_{z1} \mathbf{H}_{z1}^H+ \sigma_r^2 \mathbf{I}_M) \mathbf{u}}\geq \gamma_r. \tag{\ref{opt_w}{e}}
\end{align}
\end{subequations}
Here, all terms unrelated to $\mathbf{W}$ in constraints (\ref{opt_w}{b}) and (\ref{opt_w}{c}) are incorporated into the definitions of $\widetilde{P}_{ISAC}$ and $\widetilde{P}_{RIS}$, where $\widetilde{P}_{ISAC}$ is defined as $
\widetilde{P}_{ISAC} = \frac{1}{\xi} (P_{ISAC}^{max} - P_{ST})$, while  $\widetilde{P}_{RIS}$ is defined as $
\widetilde{P}_{RIS} = \frac{1}{\zeta} (P_{RIS}^{max} - N(P_{RP} + P_{RA})) - (\sigma_0^2 \|\mathbf{\Psi} \mathbf{h}_{rt} \mathbf{h}_{rt}^H \mathbf{\Psi}^H \|_F^2 + 2 \sigma_0^2\|\mathbf{\Psi} \|_F^2)$.

Upon examining problem (\ref{opt_w}), it is evident that the optimized variable $\mathbf{W}$ coexists with its individual components $\mathbf{w}_i$. This poses a challenge in the optimization process. To unify the representation, we consider transforming $\mathbf{W}$ into $\mathbf{w}_i$ using the relationship $\mathbf{W} = [\mathbf{w}_1,\mathbf{w}_2,\cdots,\mathbf{w}_{K+M}]$. Firstly, it is clear that constraint (\ref{opt_w}{b}) can be rewritten as:
\begin{align}
    \sum^{K+M}_{i=1}\mathbf{w}_i^H \mathbf{w}_i \leq \widetilde{P}_{ISAC},
\end{align}
which is a convex constraint. For the constraint (\ref{opt_w}{c}), we denote $\mathbf{H}_{z}\! = \!\mathbf{H}_{z1} \mathbf{H}_{z1}^H \!+\!\mathbf{H}_{z0} \mathbf{H}_{z0}^H$, so that it can be transformed into
\begin{align} \label{w_const2}
    &\sum^{K+M}_{i=1} \mathbf{w}_i^H \mathbf{H}_{z}\mathbf{w}_i
    \leq \widetilde{P}_{RIS}.
\end{align}
To address constraint (\ref{opt_w}{d}), we expand the absolute terms and rearrange the terms, resulting in the following reformulation:
\begin{align}
    \mathbf{w}_k^H \mathbf{H}_k \mathbf{w}_k \geq \frac{\tau_k}{1+\tau_k} (\sum_{i =1}^{K+M}\mathbf{w}_i^H \mathbf{H}_k \mathbf{w}_i+ c_0),
\end{align}
where we denote $\mathbf{H}_k=\mathbf{h}_k\mathbf{h}_k^H$ and the constant unrelated to $\mathbf{W}$ is written as $c_0 = \|\mathbf{h}_{r,k}^H \mathbf{\Psi}^H\|^2 \sigma_0^2 + \sigma_k^2$. The transformation of the last constraint (\ref{opt_w}{e}) follows a similar approach as that of the previous constraint, given by
\begin{align}
    \sum_{i =1}^{K+M}\mathbf{w}_i^H \mathbf{H}_u\mathbf{w}_i \geq \widetilde{\gamma}_r,
\end{align}
where $\widetilde{\gamma}_r=\gamma_r\mathbf{u}^H(\sigma_0^2 \mathbf{H}_{z0} \mathbf{H}_{z0}^H+\sigma_0^2 \mathbf{H}_{z1} \mathbf{H}_{z1}^H+ \sigma_r^2 \mathbf{I}_M) \mathbf{u}$ is a constant unrelated to $\mathbf{W}$, while we denote $\mathbf{H}_u=\mathbf{H}_t^H \mathbf{u}\mathbf{u}^H \mathbf{H}_t$. Regarding the objective function (\ref{opt_w}{a}), the sum rate-related terms can be reformulated as
\begin{align}
  \!\!\sum_{k=1}^{K} \left[2 \mathfrak{R}\left\lbrace\widetilde{\mathbf{h}}_k^H \mathbf{w}_k \right\rbrace \!\!- \!\!\mathbf{w}_k^H\mathbf{H}_{\nu, k} \mathbf{w}_k  \!\!-\!\!\sum_{i\neq k}^{K+M}\mathbf{w}_i^H\mathbf{H}_{\nu, k}\mathbf{w}_i  \right],
\end{align}
where $\widetilde{\mathbf{h}}_k=\sqrt{1+\rho_k}\nu_k \mathbf{h}_k$ and $\mathbf{H}_{\nu, k}=|\nu_k|^2\mathbf{H}_k$. The terms related to the power consumption are given as
\begin{align}
  -\mu \xi \sum^{K+M}_{i=1}\mathbf{w}_i^H \mathbf{w}_i-\mu \zeta \sum^{K+M}_{i=1} \mathbf{w}_i^H \mathbf{H}_{z}\mathbf{w}_i.
\end{align}
Consequently, the problem (\ref{opt_w}) has been transformed into an optimization problem with respect to $\mathbf{w}_i$, expressed as
\begin{subequations}\label{opt_W_2}
\begin{align} 
\max _{\mathbf{w}_i} &~  \sum_{k=1}^{K} \!\left[2 \mathfrak{R}\left\lbrace\widetilde{\mathbf{h}}_k^H \mathbf{w}_k \right\rbrace \!-\! \mathbf{w}_k^H\mathbf{H}_{\nu, k} \mathbf{w}_k  \!-\!\!\sum_{i\neq k}^{K+M}\!\mathbf{w}_i^H\mathbf{H}_{\nu, k}\mathbf{w}_i  \right] \notag\\
&-\mu \xi \sum^{K+M}_{i=1}\mathbf{w}_i^H \mathbf{w}_i-\mu \zeta \sum^{K+M}_{i=1} \mathbf{w}_i^H \mathbf{H}_{z}\mathbf{w}_i,  \tag{\ref{opt_W_2}{a}} \\ 
s.t.&~  \sum^{K+M}_{i=1}\mathbf{w}_i^H \mathbf{w}_i \leq \widetilde{P}_{ISAC}, \tag{\ref{opt_W_2}{b}} \\
& \sum^{K+M}_{i=1} \mathbf{w}_i^H \mathbf{H}_{z}\mathbf{w}_i
    \leq \widetilde{P}_{RIS},   \tag{\ref{opt_W_2}{c}} \\
&  \mathbf{w}_k^H \mathbf{H}_k \mathbf{w}_k \geq \frac{\tau_k}{1+\tau_k} (\sum_{i =1}^{K+M}\mathbf{w}_i^H \mathbf{H}_k \mathbf{w}_i+ c_0)\tag{\ref{opt_W_2}{d}}, \\
& \sum_{i =1}^{K+M}\mathbf{w}_i^H \mathbf{H}_u\mathbf{w}_i \geq \widetilde{\gamma}_r.  \tag{\ref{opt_W_2}{e}}
\end{align}
\end{subequations}
The optimization problem (\ref{opt_W_2}) remains intractable due to the non-convexity of constraints (\ref{opt_W_2}{d}) and (\ref{opt_W_2}{e}). To address this challenge, we transform the quadratic problem into an equivalent semidefinite programming (SDP) problem. Prior to this transformation,
$\widetilde{\mathbf{w}}=\begin{bmatrix}
    \mathbf{w}^T &
    1
\end{bmatrix}^T$ is introduced to incorporate the linear term and quadratic term in the objective function, which gives
\begin{align}
   & \!\!\!\!\sum_{k=1}^{K}\!\! \left[ \widetilde{\mathbf{w}}_k^H  \widetilde{\mathbf{H}}_{\nu,k,1} \widetilde{\mathbf{w}}_k \!+\!\!\sum_{i\neq k}^{K+M}\!\!\!\!\widetilde{\mathbf{w}}_i^H \widetilde{\mathbf{H}}_{\nu,k,2}  \widetilde{\mathbf{w}}_i  \right] \!\! + \!\!\!\sum^{K+M}_{i=1}\!\!\!\!\widetilde{\mathbf{w}}_i^H \widetilde{\mathbf{H}}_{zz}\widetilde{\mathbf{w}}_i , 
\end{align}
where we denote $\widetilde{\mathbf{H}}_{\nu,k,1}\!\! = \!\! \begin{bmatrix}
       -\mathbf{H}_{\nu, k} &\widetilde{\mathbf{h}}_k \\
       \widetilde{\mathbf{h}}_k^H & 0
   \end{bmatrix}$, $\widetilde{\mathbf{H}}_{\nu,k,2} = \begin{bmatrix}
       -\mathbf{H}_{\nu, k} &\mathbf{0} \\
       \mathbf{0}^T & 0
   \end{bmatrix}$ and $\widetilde{\mathbf{H}}_{zz} = \begin{bmatrix}
       -\mu \xi \mathbf{I}- \mu \zeta \mathbf{H}_{z}&\mathbf{0} \\
       \mathbf{0}^T & 0
   \end{bmatrix}$.

Accordingly, the constraints can be re-transformed as
\begin{subequations}
    \begin{align}
    & \sum^{K+M}_{i=1}\widetilde{\mathbf{w}}_i^H \widetilde{\mathbf{w}}_i \leq \widetilde{P}_{ISAC}+K+M, \\
    & \sum^{K+M}_{i=1} \widetilde{\mathbf{w}}_i^H \widetilde{\mathbf{H}}_z\widetilde{\mathbf{w}}_i
    \leq \widetilde{P}_{RIS},\\
    & \widetilde{\mathbf{w}}_k^H \widetilde{\mathbf{H}}_k \widetilde{\mathbf{w}}_k \geq \frac{\tau_k}{1+\tau_k} (\sum_{i =1}^{K+M}\widetilde{\mathbf{w}}_i^H \widetilde{\mathbf{H}}_k\widetilde{\mathbf{w}}_i+ c_0), \\   &\sum_{i=1}^{K+M}\widetilde{\mathbf{w}}_i^H \widetilde{\mathbf{H}}_u \widetilde{\mathbf{w}}_i \geq \widetilde{\gamma}_r,
\end{align}
\end{subequations}
where $\widetilde{\mathbf{H}}_z = \begin{bmatrix}
        \mathbf{H}_{z}&\mathbf{0} \\
       \mathbf{0}^T & 0
    \end{bmatrix}$, $\widetilde{\mathbf{H}}_k = \begin{bmatrix}
        \mathbf{H}_k &\mathbf{0} \\
       \mathbf{0}^T & 0
    \end{bmatrix}$, $\widetilde{\mathbf{H}}_u = \begin{bmatrix}
        \mathbf{H}_u &\mathbf{0} \\
       \mathbf{0}^T & 0
    \end{bmatrix}$. Introducing the auxiliary variable $\widetilde{\mathbf{W}}_k=\widetilde{\mathbf{w}}_k\widetilde{\mathbf{w}}_k^H$ and utilizing the properties of the trace operator, the optimization problem (\ref{opt_W_2}) can be transformed into
\begin{subequations}\label{opt_W_5}
\begin{align} 
\max _{\widetilde{\mathbf{W}}_i} &~  \!\!\sum_{k=1}^{K} \!\!\left[ \mathrm{tr}(\widetilde{\mathbf{H}}_{\nu,k,1}\widetilde{\mathbf{W}}_k) \!\!+\!\!\!\!\sum_{i\neq k}^{K+M}\!\!\mathrm{tr}(\widetilde{\mathbf{H}}_{\nu,k,2}\widetilde{\mathbf{W}}_i) \!\right] \!\!\!+\!\!\!\!\sum^{K+M}_{i=1}\!\!\!\mathrm{tr}(\widetilde{\mathbf{H}}_{zz}\widetilde{\mathbf{W}}_i)\tag{\ref{opt_W_5}{a}}\\
s.t.&~  \sum^{K+M}_{i=1}\mathrm{tr}(\widetilde{\mathbf{W}}_i) \leq \widetilde{P}_{ISAC}+K+M, \tag{\ref{opt_W_5}{b}} \\
& \sum^{K+M}_{i=1} \mathrm{tr}(\widetilde{\mathbf{H}}_z \widetilde{\mathbf{W}}_i)
    \leq \widetilde{P}_{RIS},   \tag{\ref{opt_W_5}{c}} \\
&  \mathrm{tr}(\widetilde{\mathbf{H}}_k\widetilde{\mathbf{W}}_k) \geq \frac{\tau_k}{1+\tau_k} (\sum_{i =1}^{K+M}\mathrm{tr}(\widetilde{\mathbf{H}}_k\widetilde{\mathbf{W}}_i)+ c_0), \tag{\ref{opt_W_5}{d}} \\
& \sum_{i=1}^{K+M}\mathrm{tr}(\widetilde{\mathbf{H}}_u\widetilde{\mathbf{W}}_i) \geq \widetilde{\gamma}_r,  \tag{\ref{opt_W_5}{e}} \\
    & \widetilde{\mathbf{W}}_i \succeq 0, \mathrm{rank}(\widetilde{\mathbf{W}}_i)=1. \tag{\ref{opt_W_5}{f}}
\end{align}
\end{subequations}
Due to the rank-one constraint, the optimization problem (\ref{opt_W_5}) remains non-convex. To tackle this challenge, we employ the effective method of SDR. Specifically, by removing the rank-one constraint, the problem is transformed into a standard SDP problem, which can be efficiently solved using existing toolkits such as CVX. However, the resulting $\widetilde{\mathbf{W}}^\star$ may not adhere to the rank-one constraint. To address this, we apply eigenvalue decomposition and maximum eigenvalue approximation to obtain $\mathbf{\widetilde{w}}_i$. Finally, the optimal ISAC transmitter beamformer $\mathbf{w}_i^\star$ can be obtained as
\begin{equation} \label{w_recover}
    \mathbf{w}_i^\star=\frac{1}{t_i} \mathbf{\hat{w}}_i,
\end{equation}
where $t_i = \left[ \mathbf{\widetilde{w}}_i \right]_{M+1}$, $\mathbf{\hat{w}}_i = \left[ \mathbf{\Tilde{w}}_i \right]_{1:M}$. 

\subsection{Active RIS Beamforming $\mathbf{\Psi}$ Design} \label{RIS beamformer}
In this subsection, we focus on optimizing the active RIS beamformer $\mathbf{\Psi}$. With $\mathbf{u}$ and $\mathbf{W}$ held constant, the optimization problem (\ref{formulation3}) can be simplified as
\begin{subequations}\label{opt_psi}
\begin{align} 
\max _{\mathbf{\Psi}} &~ f_2(\mathbf{\Psi}), \tag{\ref{opt_psi}{a}} \\ 
s.t.&~   \|\mathbf{\Psi}^H \mathbf{G}\mathbf{W}\|_F^2 + \|\mathbf{\Psi} \mathbf{h}_{rt} \mathbf{h}_{rt}^H \mathbf{\Psi}^H \mathbf{G} \mathbf{W} \|_F^2 \notag\\
&+ \sigma_0^2 \| \mathbf{\Psi}\mathbf{h}_{rt}\mathbf{h}_{rt}^H \mathbf{\Psi }^H \|_F^2 
  + 2\sigma_0^2\|\mathbf{\Psi}\|_F^2\leq \overline{P}_{RIS} , \tag{\ref{opt_psi}{b}} \\
& \frac{|\mathbf{h}_k^H \mathbf{w}_k|^2}{\sum_{i \neq k}^{K+M}|\mathbf{h}_k^H \mathbf{w}_i|^2 + \|\mathbf{h}_{r,k}^H \mathbf{\Psi}^H\|^2 \sigma_0^2 + \sigma_k^2} \geq \tau_k,  \tag{\ref{opt_psi}{c}} \\
& \frac{\mathbf{u}^H \mathbf{H}_t \mathbf{W W}^H \mathbf{H}_t^H \mathbf{u}}{\mathbf{u}^H(\sigma_0^2 \mathbf{H}_{z0} \mathbf{H}_{z0}^H+\sigma_0^2 \mathbf{H}_{z1} \mathbf{H}_{z1}^H+ \sigma_r^2 \mathbf{I}_M) \mathbf{u}}\geq \tau_r,  \tag{\ref{opt_psi}{d}}\\
& a_n \leq a_{max}, \forall n=1,\dots,N \tag{\ref{opt_psi}{e}},
\end{align}
\end{subequations}
where $\overline{P}_{RIS} = \frac{1}{\zeta}\left(P_{RIS}^{max} - N \left(P_{RP} + P_{RA}\right)\right)$ and the objective function $f_2(\mathbf{\Psi})$ is given in (\ref{f2}) at the top of next page.

\begin{figure*}[!t] 
\begin{align} \label{f2}
    f_2(\mathbf{\Psi}) = &\sum_{k=1}^{K} \left[
2\sqrt{1+\rho_k} \mathfrak{R}\left\lbrace \nu_k^*  \mathbf{h}_{r,k}^H \mathbf{\Psi}^H \mathbf{G} \mathbf{w}_k \right\rbrace - |\nu_k|^2 \left( \sum_{i=1}^{K+M}|\mathbf{h}_{r,k}^H \mathbf{\Psi}^H \mathbf{G} \mathbf{w}_i|^2 + \|\mathbf{h}_{r,k}^H \mathbf{\Psi}^H\|^2 \sigma_0^2 \right) \right] -\mu \zeta \left(\|\mathbf{\Psi}^H \mathbf{G}\mathbf{W}\|_F^2  \right. \notag \\
& \left. + 2\sigma_0^2\|\mathbf{\Psi}\|_F^2+ \sigma_0^2 \| \mathbf{\Psi}\mathbf{h}_{rt} \mathbf{h}_{rt}^H \mathbf{\Psi}^H \|_F^2+\|\mathbf{\Psi} \mathbf{h}_{rt} \mathbf{h}_{rt}^H \mathbf{\Psi}^H \mathbf{G} \mathbf{W} \|_F^2 \right)
\end{align}
\hrulefill
\end{figure*}
To effectively observe the problem (\ref{opt_psi}), we firstly  separate the optimization variable $\mathbf{\Psi}$ from other matrices by using the properties of the trace and norm operators. For the constraint (\ref{opt_psi}{b}), it is transformed into two quadratic and two quartic terms with respect to $\mathbf{\Psi}$ by expanding norm operation. Considering the diagonal structure of $\mathbf{\Psi}$, $\mathbf{\Psi} \mathbf{h}_{rt}$ is equivalently written as $\mathrm{diag}(\mathbf{h}_{rt})\pmb{\psi}$, where we denote $\pmb{\psi}=\mathrm{diag}(\mathbf{\Psi})$. Taking advantage of the property $\mathrm{tr}(\mathbf{X}^H \mathbf{B}\mathbf{X})=\mathbf{x}^H (\mathbf{I}\odot\mathbf{B})\mathbf{x}$ when $\mathbf{X}$ is a diagonal matrix \cite{Matrix}, the two quadratic terms can be given as
\begin{align}
\|\mathbf{\Psi}^H \mathbf{G}\mathbf{W}\|_F^2+2 \sigma_0^2\|\mathbf{\Psi}\|_F^2
\longrightarrow \pmb{\psi}^H \mathbf{A}_2 \pmb{\psi},
\end{align}
where we denote $\mathbf{A}_2 = \mathbf{I}_N \odot \overline{\mathbf{W}}+2 \sigma_0^2 \mathbf{I}_N$ and $\overline{\mathbf{W}} = \mathbf{G}\mathbf{W} \mathbf{W}^H \mathbf{G}^H$.
For the quartic terms in the constraint (\ref{opt_psi}{b}), we employ the property of trace operation $\mathrm{tr}(\mathbf{A}\mathbf{B}\mathbf{A}^H\mathbf{C})=\mathrm{vec}^H(\mathbf{A})(\mathbf{B}^T\otimes\mathbf{C})\mathrm{vec}(\mathbf{A})$ \cite{Matrix} so that we can obtain
\begin{align}
    \mathrm{vec}^H(\pmb{\psi} \pmb{\psi}^H) \mathbf{V}_2 \mathrm{vec}(\pmb{\psi} \pmb{\psi}^H),
\end{align}
where 
\begin{subequations}
    \begin{align}
        &\overline{\mathbf{W}}_t=\mathrm{diag}(\mathbf{h}_{rt}^H) \mathbf{G} \mathbf{W} \mathbf{W}^H \mathbf{G}^H  \mathrm{diag}(\mathbf{h}_{rt}),\\
        &\overline{\mathbf{H}}_{rt} = \mathrm{diag}(\mathbf{h}_{rt}^H) \mathrm{diag}(\mathbf{h}_{rt}),\\
        &\mathbf{V}_2 = \overline{\mathbf{H}}_{rt}^T \otimes \overline{\mathbf{W}}_t + \sigma_0^2 (\overline{\mathbf{H}}_{rt}^T \otimes \overline{\mathbf{H}}_{rt}).
    \end{align}
\end{subequations}
As a result, the constraint (\ref{opt_psi}{b}) can be reformulated as
\begin{align}
    \pmb{\psi}^H \mathbf{A}_2 \pmb{\psi}+ \mathrm{vec}^H(\pmb{\psi} \pmb{\psi}^H) \mathbf{V}_2 \mathrm{vec}(\pmb{\psi} \pmb{\psi}^H)  \leq \overline{P}_{RIS}.
\end{align}
Similarly, with the relationship of $\mathbf{\Psi} \mathbf{h}_{r,k} = \mathrm{diag}(\mathbf{h}_{r,k}) \pmb{\psi}$, the constraint (\ref{opt_psi}{c}) can be readily transformed as
\begin{align}
 &\pmb{\psi}^H\overline{\mathbf{W}}_{k,k} \pmb{\psi} \geq \frac{\tau_k}{1+\tau_k}(\pmb{\psi}^H\overline{\mathbf{W}}_k \pmb{\psi} +\sigma_0^2 \pmb{\psi}^H\overline{\mathbf{H}}_{r,k}\pmb{\psi}+\sigma_k^2) ,
\end{align} 
where $\overline{\mathbf{W}}_{k,k}\!=\!\mathrm{diag}(\mathbf{h}_{r,k}^H)\mathbf{G}\mathbf{w}_k\mathbf{w}_k^H \mathbf{G}^H\mathrm{diag}(\mathbf{h}_{r,k})$,$\overline{\mathbf{W}}_k\!=\!\mathrm{diag}(\mathbf{h}_{r,k}^H)\mathbf{G}\mathbf{W}\mathbf{W}^H \mathbf{G}^H\mathrm{diag}(\mathbf{h}_{r,k})$ and $\overline{\mathbf{H}}_{r,k} \!=\!\! \mathrm{diag}(\mathbf{h}_{r,k}^H)\mathrm{diag}(\mathbf{h}_{r,k})$.
The constraint (\ref{opt_psi}{d}) contains quartic terms, quadratic terms, and constant terms, which can be transformed in a similar manner and rewritten as
\begin{align}
   \pmb{\psi}^H\overline{\mathbf{W}}_t \pmb{\psi}\pmb{\psi}^H\overline{\mathbf{U}}_t\pmb{\psi} \!\geq \! \tau_r \sigma_0^2(\pmb{\psi}^H\overline{\mathbf{H}}_{rt}  \pmb{\psi}\pmb{\psi}^H\overline{\mathbf{U}}_t \pmb{\psi}\!+\!\pmb{\psi}^H\overline{\mathbf{G}}\pmb{\psi}) \!+\!\sigma_r^2\tau_r,
\end{align}
where we denote $\overline{\mathbf{W}}_t\!=\!\!\mathrm{diag}(\mathbf{h}_{rt}^H)\mathbf{G}\mathbf{W}\mathbf{W}^H \mathbf{G}^H\mathrm{diag}(\mathbf{h}_{rt}) $, $\overline{\mathbf{U}}_t\!=\!\mathrm{diag}(\mathbf{h}_{rt}^H)\mathbf{G}\mathbf{u}\mathbf{u}^H \mathbf{G}^H\mathrm{diag}(\mathbf{h}_{rt}) $ and $\overline{\mathbf{H}}_{rt}\! = \!\mathrm{diag}(\mathbf{h}_{rt}^H)\mathrm{diag}(\mathbf{h}_{rt})$, $\overline{\mathbf{G}}\! \!=\!\! \mathbf{G}\mathbf{G}^H \!\!\odot \!\!\mathbf{I}$.
By integrating the quartic terms, constraint (\ref{opt_psi}{d}) is further transformed to
\begin{align}
    \mathrm{vec}^H(\pmb{\psi} \pmb{\psi}^H) \mathbf{V}_3 \mathrm{vec}(\pmb{\psi} \pmb{\psi}^H)+ \pmb{\psi}^H \mathbf{A}_3 \pmb{\psi}+ \sigma_r^2\tau_r \leq 0,
\end{align}
where we denote $\mathbf{V}_3=\tau_r \sigma_0^2 (\overline{\mathbf{H}}_{rt}^T \otimes \overline{\mathbf{U}}_t)-(\overline{\mathbf{W}}_{t}^T \otimes \overline{\mathbf{U}}_t)$ and $\mathbf{A}_3=\tau_r \sigma_0^2 \overline{\mathbf{G}}$.
To address the objective function (\ref{opt_psi}{a}), utilizing the notations defined earlier, we can integrate the quartic terms and the quadratic terms separately, which allows us to express it in a compact form as 
\begin{align}
   & 2\mathfrak{R} \left\lbrace \mathbf{b}_1^H \pmb{\psi} \right\rbrace -\pmb{\psi}^H \mathbf{A}_1 \pmb{\psi}-\mathrm{vec}^H(\pmb{\psi} \pmb{\psi}^H) \mathbf{V}_1 \mathrm{vec}(\pmb{\psi} \pmb{\psi}^H),
\end{align}
where 
\begin{subequations} \label{nota4}
\begin{align}
    {\mathbf{b}_1} & {= \sum_{k=1}^{K}\sqrt{1+\rho_k} \nu_k^* \mathrm{diag}(\mathbf{h}_{r,k})\mathbf{G} \mathbf{w}_k} , \tag{\ref{nota4}{a}}\\
     \mathbf{A}_1 &= \sum_{k=1}^{K}|\nu_k|^2 \left[\overline{\mathbf{W}}_k+\sigma_0^2 \overline{\mathbf{H}}_{r,k}\right]+ \mu \zeta \mathbf{A}_2 ,\tag{\ref{nota4}{b}}\\
     \mathbf{V}_1 &= \mu \zeta \mathbf{V}_2.\tag{\ref{nota4}{c}}
\end{align}    
\end{subequations}
To facilitate a more compact expression, we define a vector $\mathbf{x}$ as the vectorization of the outer product of $\pmb{\psi}$ and its Hermitian transpose, i.e., $\mathbf{x} = \mathrm{vec}(\pmb{\psi}\pmb{\psi}^H)$. This allows us to recast the problem (\ref{opt_psi}) in terms of $\pmb{\psi}$ as
\begin{subequations}\label{opt_psi2}
\begin{align} 
\max _{\pmb{\psi}} &~  2\mathfrak{R} \left\lbrace \mathbf{b}_1^H \pmb{\psi} \right\rbrace- \pmb{\psi}^H \mathbf{A}_1\pmb{\psi} - \mathbf{x}^H \mathbf{V}_1\mathbf{x}, \tag{\ref{opt_psi2}{a}} \\ 
s.t.&~  \pmb{\psi}^H \mathbf{A}_2 \pmb{\psi}+ \mathbf{x}^H \mathbf{V}_2 \mathbf{x}  \leq \overline{P}_{RIS},  \tag{\ref{opt_psi2}{b}} \\
&\pmb{\psi}^H\overline{\mathbf{W}}_{\!k,k} \pmb{\psi} \!\geq\! \frac{\tau_k}{1\!+\!\tau_k} \!(\pmb{\psi}^H\overline{\mathbf{W}}_k \pmb{\psi} \!+\!\sigma_0^2 \pmb{\psi}^H\overline{\mathbf{H}}_{r,k}\pmb{\psi}\!+\!\sigma_k^2) , \tag{\ref{opt_psi2}{c}} \\
& \mathbf{x}^H\mathbf{V}_3 \mathbf{x}+ \pmb{\psi}^H \mathbf{A}_3 \pmb{\psi}+ \sigma_r^2\tau_r \leq 0,  \tag{\ref{opt_psi2}{d}}\\
& |\pmb{\psi}_n| \leq a_{max}, \forall n=1,\dots,N ,\tag{\ref{opt_psi2}{e}}
\end{align}
\end{subequations}
The optimization problem (\ref{opt_psi2}) poses a significant challenge due to two primary factors. Firstly, the problem involves a combination of quartic, quadratic, and linear terms with respect to $\pmb{\psi}$, making it difficult to solve directly using conventional optimization techniques. Secondly, the constraint (\ref{opt_psi2}{c}) takes the form "Convex $\geq$ Convex," rendering it non-convex. To address the first challenge, we consider employing the MM framework, which involves constructing a surrogate function to approximate the high-order terms. To construct such a surrogate function, we first introduce the following lemma:
\begin{lemma}
For any vector $\mathbf{x}$ and any Hermitian matrix $\mathbf{A}$, the following inequality holds:
\begin{align}
    \!\!\!\mathbf{x}^H \mathbf{A} \mathbf{x} &\!\leq \!\lambda \mathbf{x}^H \mathbf{x}\!+ \!\mathbf{x}_s^H (\lambda\mathbf{I}\!-\!\mathbf{A})\mathbf{x}_s \!+\! 2 \mathfrak{R}\left\lbrace \mathbf{x}^H(\mathbf{A}\!-\!\lambda\mathbf{I})\mathbf{x}_s\right\rbrace\!,
\end{align}
where $\mathbf{x}_s$ is a fixed point,  while $\lambda$ represents the maximum eigenvalue of the matrix $\mathbf{A}$.
\label{lm-1}
\end{lemma}
\begin{proof}
Please refer the derivation to Appendix \ref{app_a}.
\end{proof}
To address the quartic term $\mathbf{x}^H \mathbf{V}_1 \mathbf{x}$ in the objective function (\ref{opt_psi2}{a}), we can leverage the surrogate function introduced in Lemma \ref{lm-1}. The surrogate function for this term is given as
\begin{align} \label{surr1}
    \!\!\!\lambda_{v1} \mathbf{x}^H \mathbf{x} \!+\! 2 \mathfrak{R}\lbrace \mathbf{x}^H (\mathbf{V}_1 \!-\!\lambda_{v1}\mathbf{I})\mathbf{x}_s\rbrace \!+\! \mathbf{x}_s^H  (\lambda_{v1}\mathbf{I}\!-\!\mathbf{V}_1)  \mathbf{x}_s,
\end{align}
where $\lambda_{v1}$ is the maximum eigenvalue of $\mathbf{V}_1$, and the fixed point $\mathbf{x}_s$ is obtained by selecting the optimized result $\mathbf{x}$ from the last iteration step. The term $\mathbf{x}^H \mathbf{x}$ in (\ref{surr1}) remains quartic with respect to $\pmb{\psi}$. Since $|\pmb{\psi}_n| \leq a_{max}$, an upper bound for $\mathbf{x}^H \mathbf{x}$ can be derived, leading to
\begin{align}
    \mathbf{x}^H \mathbf{x}&=\mathrm{vec}^H(\pmb{\psi}\pmb{\psi}^H)\mathrm{vec}(\pmb{\psi}\pmb{\psi}^H)=(\pmb{\psi}^* \otimes \pmb{\psi})^H (\pmb{\psi}^* \otimes \pmb{\psi})\notag\\
    &=(\pmb{\psi}^T \pmb{\psi}^*)\otimes (\pmb{\psi}^H \pmb{\psi})=  \|\pmb{\psi}\|^4 \leq N^2 a_{max}^4.
\end{align}
The upper bound becomes tight when  all the  active RIS elements reach their maximum amplification limit $a_{max}$. Consequently, the final surrogate function of the objective function (\ref{opt_psi2}{a}) can be expressed as
\begin{align} \label{surr2}
    \!\!\!\!\!\!\lambda_{v1} N^2 a_{max}^4 \!\!+\! 2 \mathfrak{R}\lbrace \mathbf{x}^H (\mathbf{V}_1 \!-\!\lambda_{v1}\mathbf{I})\mathbf{x}_s\!\rbrace \!+\! \mathbf{x}_s^H  (\lambda_{v1}\mathbf{I}\!-\!\mathbf{V}_1\!)  \mathbf{x}_s.\!
\end{align}
By omitting the constant term and leveraging the properties of vectorization, the objective function can be reformulated as
\begin{align}
    2\mathfrak{R} \left\lbrace \mathbf{b}_1^H \pmb{\psi} \right\rbrace- \pmb{\psi}^H \mathbf{A}_1\pmb{\psi} -  \pmb{\psi}^H (\overline{\mathbf{X}}_{s1}+\overline{\mathbf{X}}_{s1}^H) \pmb{\psi},
\end{align}
where we denote $\overline{\mathbf{x}}_{s1}\!\!=\!(\mathbf{V}_1 -\lambda_{v1}\mathbf{I})\mathbf{x}_s$ and $\overline{\mathbf{X}}_{s1} =\mathrm{unvec}(\overline{\mathbf{x}}_{s1})$.

With similar processes, the constraints  (\ref{opt_psi2}{b}) and (\ref{opt_psi2}{d}) can be transformed respectively as
\begin{subequations}\label{nota_psi_5}
\begin{align}
    &\pmb{\psi}^H (\overline{\mathbf{X}}_{s2}+\overline{\mathbf{X}}_{s2}^H) \pmb{\psi}+\pmb{\psi}^H \mathbf{A}_2\pmb{\psi} +c_2 \leq \overline{P}_{RIS}, \tag{\ref{nota_psi_5}{a}}\\
    &\pmb{\psi}^H (\overline{\mathbf{X}}_{s3}+\overline{\mathbf{X}}_{s3}^H) \pmb{\psi}+\pmb{\psi}^H \mathbf{A}_3\pmb{\psi}+\sigma_r^2\tau_r +c_3  \leq 0 ,\tag{\ref{nota_psi_5}{b}}
\end{align}  
\end{subequations}
where 
\begin{subequations}
    \begin{align}
\overline{\mathbf{X}}_{s2}&=\mathrm{unvec}[(\mathbf{V}_2 -\lambda_{v2}\mathbf{I})\mathbf{x}_s],\\
    c_2 &= \mathbf{x}_s^H  (\lambda_{v2}\mathbf{I}-\mathbf{V}_2)  \mathbf{x}_s + \lambda_{v2} N^2 a_{max}^4,\\   \overline{\mathbf{X}}_{s3}&=\mathrm{unvec}[(\mathbf{V}_{3}-\lambda_{v3}\mathbf{I}_N)\mathbf{x}_s],\\
    c_3 &= \mathbf{x}_{s}^H(\lambda_{v3}\mathbf{I}_N-\mathbf{V}_{3}) \mathbf{x}_{s}+\lambda_{v3} N^2 a_{max}^4.
\end{align}
\end{subequations}
It is noteworthy that determining the maximum eigenvalues $\lambda_{v1}$, $\lambda_{v2}$, and $\lambda_{v3}$ can be a computationally demanding task, particularly for large values of $N$, as the matrices $\mathbf{V}_1$, $\mathbf{V}_2$, and $\mathbf{V}_3$ have dimensions of $N^2 \times N^2$. To alleviate the computational burden associated with finding the eigenvalues, we propose the following proposition:
\begin{proposition}
The maximum eigenvalue $\lambda_{v1}$ of $\mathbf{V}_1$, $\lambda_{v2}$ of $\mathbf{V}_2$ and $\lambda_{v3}$ of $\mathbf{V}_3$ are given as
\begin{subequations} \label{surr_eigen}
    \begin{align}
    \lambda_{v1} &= \mu \zeta \lambda_{v2}, \\
    \lambda_{v2} &\approx \lambda_w \max\lbrace \overline{\mathbf{H}}_{rt} \rbrace, \\
    \lambda_{v3} &= \lambda_{t} \rm\textbf{\rm tr}(\overline{\mathbf{U}}_t) ,
\end{align}
\end{subequations}
where $\lambda_w$ is the maximum eigenvalue of $ \overline{\mathbf{W}}_{t2} = \mathbf{W}^H \mathbf{G}^H  \textbf{\rm diag}(\mathbf{h}_{rt})\textbf{\rm diag}(\mathbf{h}_{rt}^H) \mathbf{G} \mathbf{W} \in \mathbb{C}^{(M+K) \times (M+K)}$, and $\lambda_{t}$ is the maximum eigenvalue of $\tau_r \sigma_0^2\overline{\mathbf{H}}_{rt} - \overline{\mathbf{W}}_{t} \in \mathbb{C}^{N \times N}$.
\label{prop-1}
\end{proposition}
\begin{proof}
Please refer the derivation to Appendix \ref{app_b}.
\end{proof}
By substituting the surrogate functions for the quartic terms, the optimization problem (\ref{opt_psi2}) can be transformed into
\begin{subequations}\label{opt_psi3}
\begin{align} 
\max _{\pmb{\psi}} &~ 2\mathfrak{R} \left\lbrace \mathbf{b}_1^H \pmb{\psi} \right\rbrace- \pmb{\psi}^H \mathbf{A}_1\pmb{\psi} -  \pmb{\psi}^H (\overline{\mathbf{X}}_{s1}+\overline{\mathbf{X}}_{s1}^H) \pmb{\psi}, \tag{\ref{opt_psi3}{a}} \\ 
s.t.&~  \pmb{\psi}^H (\overline{\mathbf{X}}_{s2}+\overline{\mathbf{X}}_{s2}^H) \pmb{\psi}+\pmb{\psi}^H \mathbf{A}_2\pmb{\psi} +c_2 \leq \overline{P}_{RIS}, \tag{\ref{opt_psi3}{b}} \\
& \!\!\!\pmb{\psi}^H\overline{\mathbf{W}}_{k,k} \pmb{\psi} \!\geq \!\frac{\tau_k}{1\!+\!\tau_k}\!(\pmb{\psi}^H\overline{\mathbf{W}}_k \pmb{\psi} \!+\!\sigma_0^2 \pmb{\psi}^H\overline{\mathbf{H}}_{r,k}\pmb{\psi}\!+\!\sigma_k^2),\!   \tag{\ref{opt_psi3}{c}} \\
& \pmb{\psi}^H (\overline{\mathbf{X}}_{s3}+\overline{\mathbf{X}}_{s3}^H) \pmb{\psi}+\pmb{\psi}^H \mathbf{A}_3\pmb{\psi}+\sigma_r^2\tau_r +c_3  \leq 0,  \tag{\ref{opt_psi3}{e}}\\
& |\pmb{\psi}_{n}| \leq a_{max}, \forall n=1,\dots,N .\tag{\ref{opt_psi3}{d}}
\end{align}
\end{subequations}
The optimization problem (\ref{opt_psi3}) is a non-convex quadratic constrained quadratic programming (QCQP) problem due to the non-positive semidefiniteness of $\overline{\mathbf{X}}_{s1}\!+\!\overline{\mathbf{X}}_{s1}^H$ and the non-convexity of  constraint (\ref{opt_psi3}{c}). To address this non-convexity, we employ SDR approach. By letting $\overline{\pmb{\psi}} =\begin{bmatrix}
    \pmb{\psi}^T & 1
\end{bmatrix}^T\!$ and $\overline{\mathbf{\Psi}} \!=\!\overline{\pmb{\psi}} \overline{\pmb{\psi}}^H$, 
the problem can be transformed as
\begin{subequations}\label{opt_psi5}
\begin{align} 
\max _{\overline{\mathbf{\Psi}} } &~  \mathrm{tr}(\overline{\mathbf{A}}_1\overline{\mathbf{\Psi}}), \tag{\ref{opt_psi5}{a}} \\ 
s.t.&~  \mathrm{tr}(\overline{\mathbf{A}}_2\overline{\mathbf{\Psi}}) +c_2 \leq \overline{P}_{RIS}, \tag{\ref{opt_psi5}{b}} \\
& \mathrm{tr}( \overline{\mathbf{A}}_{k,k}\overline{\mathbf{\Psi}}) \geq \frac{\tau_k}{1+\tau_k}(\mathrm{tr}( \overline{\mathbf{A}}_{k}\overline{\mathbf{\Psi}} ) +\sigma_k^2),   \tag{\ref{opt_psi5}{c}} \\
&  \mathrm{tr}( \overline{\mathbf{A}}_{3} 
 \overline{\mathbf{\Psi}})+\sigma_r^2\tau_r +c_3  \leq 0,  \tag{\ref{opt_psi5}{d}}\\
& \mathbf{\Psi}_{n,n} \leq a_{max}^2, \forall n=1,\dots,N, \mathbf{\Psi}_{N+1,N+1}=1 , \tag{\ref{opt_psi5}{e}}\\
&\overline{\mathbf{\Psi}} \succeq 0, \mathrm{rank}(\overline{\mathbf{\Psi}})=1, \tag{\ref{opt_psi5}{f}}
\end{align}
\end{subequations}
where we denote
\begin{subequations}
    \begin{align}
    &\overline{\mathbf{A}}_1 = \begin{bmatrix}
        -(\mathbf{A}_1+\overline{\mathbf{X}}_{s1}+\overline{\mathbf{X}}_{s1}^H) & \mathbf{b}_1\\
       \mathbf{b}_1^H & 0\\
    \end{bmatrix},\\
    &\overline{\mathbf{A}}_2 = \begin{bmatrix}
        -(\mathbf{A}_2+\overline{\mathbf{X}}_{s2}+\overline{\mathbf{X}}_{s2}^H) & \mathbf{0}\\
       \mathbf{0}^T & 0\\
    \end{bmatrix},\\
    &\overline{\mathbf{A}}_{k,k} = \begin{bmatrix}
        \overline{\mathbf{W}}_{k,k} & \mathbf{0}\\
       \mathbf{0}^T & 0\\
    \end{bmatrix}, 
    \overline{\mathbf{A}}_{k} = \begin{bmatrix}
    \overline{\mathbf{W}}_{k}+\sigma_0^2 \overline{\mathbf{H}}_{r,k} & \mathbf{0}\\
       \mathbf{0}^T & 0\\
    \end{bmatrix},\\
    &\overline{\mathbf{A}}_{3} = \begin{bmatrix}
(\mathbf{A}_3+\overline{\mathbf{X}}_{s3}+\overline{\mathbf{X}}_{s3}^H) & \mathbf{0}\\
       \mathbf{0}^T & 0\\
    \end{bmatrix}.
\end{align}
\end{subequations}
By eliminating rank-one constraint, the problem (\ref{opt_psi5}) transforms into an SDP problem {\cite{boyd}}, which can be efficiently solved using existing convex optimization toolkits such as SDPT3 {\cite{grant}}. Once the optimal solution $\overline{\mathbf{\Psi}} $ is obtained, various rank-one decomposition methods, such as eigenvalue decomposition and Gaussian randomization \cite{SDR}, can be employed to recover $\pmb{\psi}$. Subsequently, amplification $\mathbf{A}$ and phase shift $\mathbf{\Theta}$ are determined using $\mathbf{A}\!=\!|\mathrm{diag}(\pmb{\psi})|$  and $\mathbf{\Theta}\!=\!\angle[\mathrm{diag}(\pmb{\psi})]$, respectively.

\subsection{Convergence and Complexity Analysis}
The proposed alternative optimization algorithm is summarized in \textbf{Algorithm 1}, which leverages the solutions developed in Sections \ref{pre-processing}, B, C, and D. Algorithm 1 iteratively solves problems (\ref{opt_u}), (\ref{opt_W_5}), and (\ref{opt_psi5}) to obtain the optimal solutions of $\mathbf{W}$, $\mathbf{u}$, and $\mathbf{\Psi}$, respectively. This process continues until the convergence criterion is satisfied, i.e., the change in the value of $\eta$ becomes stable.

\subsubsection{Convergence Analysis}
Given the iterative nature of \textbf{Algorithm 1}, a convergence analysis is crucial to establish its effectiveness. Let ($\mathbf{u}^t$, $\widetilde{\mathbf{W}}_i^t$, $\overline{\mathbf{\Psi}}^t$) denote the obtained feasible results in the $t$-th iteration. The iterative optimization routine is given as $\cdots \longrightarrow$ ($\mathbf{u}^t$, $\widetilde{\mathbf{W}}_i^t$, $\overline{\mathbf{\Psi}}^t$) $\longrightarrow$ ($\mathbf{u}^{t+1}$, $\widetilde{\mathbf{W}}_i^t$, $\overline{\mathbf{\Psi}}^t$)$\longrightarrow$ ($\mathbf{u}^{t+1}$, $\widetilde{\mathbf{W}}_i^{t+1}$, $\overline{\mathbf{\Psi}}^t$)$\longrightarrow$ ($\mathbf{u}^{t+1}$, $\widetilde{\mathbf{W}}_i^{t+1}$, $\overline{\mathbf{\Psi}}^{t+1}$) $\longrightarrow \cdots $. Since the generalized Rayleigh quotient optimization problem (\ref{opt_u}) is constrained by the maximum eigenvalue according to the Rayleigh quotient inequality, it is evident that solving problem (\ref{opt_u}) is a non-decreasing process in the sense that
\begin{align}
    \eta (\mathbf{u}^{t+1}, \widetilde{\mathbf{W}}_i^t, \overline{\mathbf{\Psi}}^t) \geq \eta (\mathbf{u}^t, \widetilde{\mathbf{W}}_i^t, \overline{\mathbf{\Psi}}^t).
\end{align}
On the other hand, the problem (\ref{opt_W_5}) is convex with respect to $\mathbf{W}_i$ and is a maximization process, so that 
\begin{align}
    \eta (\mathbf{u}^{t+1}, \widetilde{\mathbf{W}}^{t+1}_i, \overline{\mathbf{\Psi}}^t) \geq \eta (\mathbf{u}^{t+1}, \widetilde{\mathbf{W}}_i^t, \overline{\mathbf{\Psi}}^t).
\end{align}
Similarly, since $\overline{\mathbf{\Psi}}^{t+1}$ is the optimal solution of the problem (\ref{opt_psi5}), one has
\begin{align}
    \eta (\mathbf{u}^{t+1}, \widetilde{\mathbf{W}}^{t+1}_i, \overline{\mathbf{\Psi}}^{t+1}) \geq \eta (\mathbf{u}^{t+1}, \widetilde{\mathbf{W}}^{t+1}_i, \overline{\mathbf{\Psi}}^t).
\end{align}
Therefore, the objective function $\eta$ is monotonically non-decreasing in each iteration. Furthermore, since the spectrum and energy resources are constrained, the energy efficiency $\eta$ is bounded. Consequently, the proposed algorithm can converge to a stationary point. It is worth mentioning that the Gaussian randomization or maximum eigenvalue approximation for rank-one solution recovery may cause performance loss and thus hard to obtain a global optimal solution to the problem, but alternative optimization is guaranteed to converge to a solution
that satisfies the Karush-Kuhn-Tucker (KKT) conditions, which will be displayed in Section \ref{simulation}.

\begin{algorithm}[!t] 
\caption{Proposed alternative optimization algorithm} 
    \begin{algorithmic}[1] 
        \REQUIRE ~ 
        Channels ${\mathbf{G}}$,  $\mathbf{h}_{r,k}$, $\mathbf{h}_{rt}$. Noise power $\sigma^2_0$, $\sigma^2_k$, $\sigma^2_r$, maximum amplification coefficient $a_{max}$, power budget $P_{ISAC}^{max}$, $P_{RIS}^{max}$, user SINR and echo SNR threshold $\tau_k$, $\tau_r$
        \ENSURE ~ 
        Optimized variables $\mathbf{W}$, $\mathbf{A}$, $\mathbf{\Theta}$, ${\mathbf{u}}$ and the energy efficiency value $\eta$	
        \STATE Initialize transmit beamformer $\mathbf{W}$, echo receiving beamformer $\mathbf{u}$ and active RIS matrix $\mathbf {\Psi}$;
        \WHILE {no convergence of $\eta$}
        \STATE Obtain fractional decoupling variable via $\mu^\star=\eta=\frac{R_b}{P}$;
        \STATE Update auxiliary variables ${\pmb{\rho}}$ and ${\pmb{\nu}}$ by (\ref{auxi1}) and (\ref{auxi2});
        \STATE Update ${\mathbf{u}}$ by calculating the eigenvector of $\mathbf{D}^{-1}\mathbf{C}$ associated with the maximum eigenvalue;
        \STATE Update ${\mathbf{W}}$ by solving problem (\ref{opt_W_5}) and (\ref{w_recover});
        \STATE Update ${\mathbf{\Psi}}$ by solving problem (\ref{opt_psi5}) and performing rank-one decomposition;
        \ENDWHILE	
        \STATE Obtain $\mathbf{A}$ and $\mathbf{\Theta}$ from ${\mathbf \Psi}$;
        \RETURN Optimized $\mathbf{W}$, $\mathbf{A}$, $\mathbf{\Theta}$ and energy efficiency $\eta^\star$. 
    \end{algorithmic}
\end{algorithm}

\subsubsection{Complexity Analysis}
The complexity is another critical aspect  of the algorithm.  For each iteration, the complexities of obtaining $\mathbf{u}$, $\mathbf{W}$, $\mathbf{\Psi}$ are given as follows:
\begin{itemize}
    \item The optimization of $\mathbf{u}$ involves an eigenvalue decomposition process to determine the optimal solution $\mathbf{u}$ of the Rayleigh quotient. This process contributes to the complexity of optimizing $\mathbf{u}$, which is $\mathcal{O}\left(M^3\right)$ \cite{RIS_ISAC_PD}.
    \item The optimization of $\mathbf{W}$ primarily involves solving the SDP problem (\ref{opt_W_5}) and performing the eigenvalue decomposition in (\ref{w_recover}). The complexities arising from other components are negligible compared to the SDP problem and the eigenvalue decomposition. Utilizing the interior point algorithm, the SDP problem (\ref{opt_W_5}), after omitting the rank-one constraint, can be solved with a complexity of $\mathcal{O}\left((M+1)^{4.5} \log(1/\epsilon)\right)\simeq\mathcal{O}\left(M^{4.5} \log(1/\epsilon)\right)$, where $\epsilon$ denotes the desired solution accuracy \cite{SDR}. The eigenvalue decomposition for recovering the rank-one result has a complexity of $\mathcal{O}\left(K(M+1)^3\right)\simeq\mathcal{O}\left(K M^3\right)$.
    \item The optimization of $\mathbf{\Psi}$ involves three primary sources of complexity: finding the eigenvalues for constructing the surrogate functions, solving the SDP problem (\ref{opt_psi5}), and performing the rank-one decomposition. Determining the eigenvalues $\lambda_{v2}$ and $\lambda_{v3}$ using (\ref{surr_eigen}{b}) and (\ref{surr_eigen}{c}) incurs complexities of $\mathcal{O}\left((K+M)^3\right)$ and $\mathcal{O}\left(N^3\right)$, respectively, due to the eigenvalue decomposition. Solving the SDP problem (\ref{opt_psi5}) after omitting the rank-one constraint requires a complexity of $\mathcal{O}\left((N+1)^{4.5} \log(1/\epsilon)\right)\simeq\mathcal{O}\left(N^{4.5} \log(1/\epsilon)\right)$. Finally, recovering $\mathbf{\Psi}$ from $\overline{\mathbf{\Psi}}$ through rank-one decomposition involves a complexity of $\mathcal{O}\left((N+1)^3\right)\simeq\mathcal{O}\left(N^3\right)$.
\end{itemize}

Overall, the computational complexity of each iteration for the \textbf{Algorithm 1} is $\mathcal{O}\left(M^3+M^{4.5} \log(1/\epsilon)+K M^3+(K+M)^3+N^3+N^{4.5}\right.$ 

\noindent $\left. \log(1/\epsilon)+N^3\right)$. 

\section{Simulation Results}\label{simulation}
In this section, we present simulation results to evaluate the effectiveness of the proposed algorithm for maximizing EE in an active RIS-assisted ISAC system. Unless otherwise specified, the BS is equipped with an $8$-antenna ULA array, while the active RIS consists of $32$ reflecting elements. The active RIS amplification is limited to $a_{max}\!=\!10$. 
It is set that BS and active RIS are located at ($0$m, $0$m) and ($250$m, $0$m), respectively.
The system serves $K\!=\!4$ downlink users randomly distributed within a zone centered at ($200$m, $30$m), while the target is located at ($200$m, $-50$m). Additionally, the total system power consumption is restricted to $P\!=\!10$ dBm, with $90\%$ allocated to  BS and the remaining $10\%$ to active RIS. The inverse of the energy conversion coefficients at BS and active RIS are set to $\xi\!=\!\zeta\!=\!1.1$. The static hardware power $P_{ST}$ in BS is set to $-10$  dBm, while hardware power of the amplifier and phase tuner at active RIS are set to $P_{RP}\!=\!P_{RA}\!=\!-30$ dBm. All noise power levels are set to $\sigma_0^2\!=\!\sigma_k^2\!=\!\sigma_r^2\!=\!-80$ dBm. For simplicity, the user SINR threshold $\tau_k$ is set to $10$ dB for all users, while echo SNR threshold is set to $0$ dB. {All the communication channels are modeled as Rician channels according to \cite{RIS_ISAC_jiang}, i.e.
\begin{align}
    \mathbf{H}\! = \!\sqrt{\beta}\left(\sqrt{\frac{\kappa}{\kappa+1}}\mathbf{a}(\upsilon_{r})\mathbf{a}^H(\upsilon_{t})\!+\!\sqrt{\frac{1}{\kappa+1}}\mathbf{H}_{NLoS}\right).
\end{align}
We set the Rician factor $\kappa$ to $10$. The path loss, denoted by $\sqrt{\beta}$, is modeled based on 3GPP standards \cite{threeGPP} \cite{channel}, which is given as $pl_{\text{dB}} = 35.6 + 22.0 \log_{10} d$ (dB) associated with the distance $d$. The noise power spectral density is set to $n_p = -76$ dBm/Hz so that $\beta = 10^{\frac{-n_p - pl_{\text{dB}}}{10}}$.} In addition, the symbol $\mathbf{a}(\upsilon_{t})$ and $\mathbf{a}(\upsilon_{r})$ are the steering vectors associated with the transmit direction $\upsilon_{t}$ and receiving direction $\upsilon_{r}$ and $\mathbf{H}_{NLoS}$ is non-line-of-sight component of the channel, which follows Rayleigh distribution of zero mean and unit variance. Moreover, the sensing channel $\mathbf{h}_{rt}$ is typically modeled as the steering vectors associated with the target direction similar to \cite{Joint_Transmit_Beamforming}, \cite{senchan}. Additionally, the transmit beamformer $\mathbf{W}$, receiving beamformer $\mathbf{u}$ and RIS phase shift $\boldsymbol{\Psi}$ are randomly initialized and the simulations are conducted with $100$ Monte Carlo realizations to avoid the effect of randomness and initialization.

To evaluate the performance of the proposed alternating optimization algorithm, three benchmark approaches are employed for comparisons. The first benchmark is the EE optimization in the passive RIS-assisted ISAC system, presented in \cite{PassiveRISDFRCTGCN}. The second benchmark is the SE optimization in the active RIS-assisted ISAC scenario, based on the work in \cite{ARIS_ISAC_secure}. However, the objective function in \cite{ARIS_ISAC_secure} is secrecy rate, which differs from our objective function. To ensure a fair comparison, the objective function in \cite{ARIS_ISAC_secure} is modified to be SE. The final benchmark is the no optimization case, where all the optimizing variables are set randomly. In addition, simulation parameters such as system power consumption are set the same as those in the active RIS-assisted ISAC scenario for fairness. For convenient expression in simulation figures, our proposed system  is termed as ``Active EE'', the first benchmark is termed as ``Passive EE'', the second benchmark is termed as ``SE'', and the final benchmark is termed as ``No optimization'', respectively.
\begin{figure}[!t]
    \centering
    \includegraphics[scale=0.45]{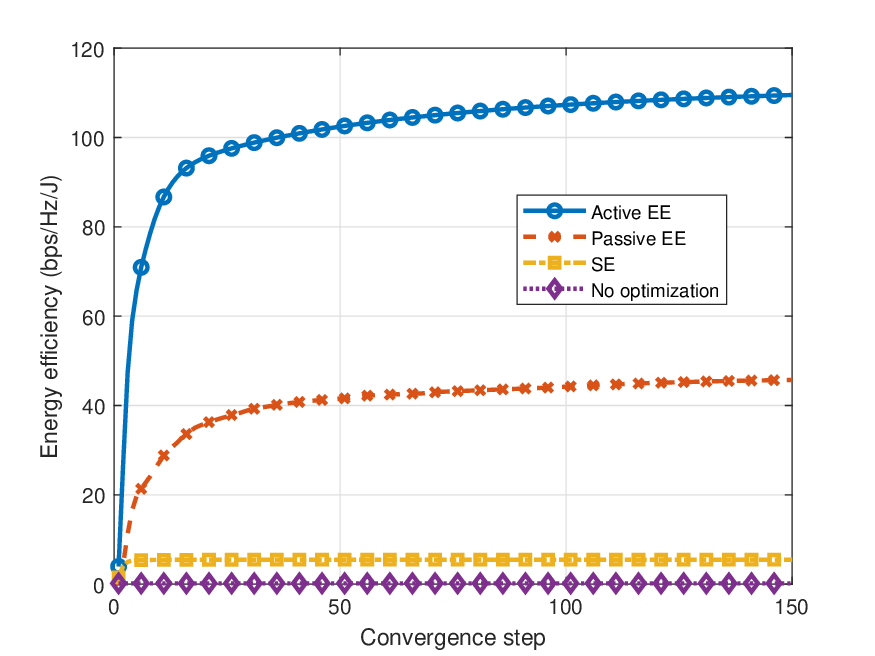}
    \caption{Convergence behaviour of the proposed alternating optimization
algorithm: EE versus the convergence step.}
    \label{fig:Energy efficiency}
\end{figure}

To assess the EE convergence performance of our proposed alternating optimization algorithm for the active RIS-assisted ISAC architecture, we compare it with three benchmark scenarios. Fig. \ref{fig:Energy efficiency} depicts the changes of EE during the convergence process. On the other, it is worth reminding that EE is a quotient of SE and energy consumption, i.e $\eta = \frac{R_b}{P}$ shown in (\ref{EE_concept}). Hence, we provide insights into the interplay between SE and energy consumption during convergence in Figs. \ref{fig:Spectrum efficiency} and \ref{fig:Energy Consumption} for understanding the changes of EE. As shown in Fig. \ref{fig:Energy efficiency}, the EE optimization process exhibits excellent convergence behavior, validating the effectiveness of our proposed algorithm. Furthermore, the ``Active EE'' case achieves higher EE compared to the ``Passive EE'' case due to the active RIS's ability to mitigate the multiplicative fading effect. Notably, the proposed scheme demonstrates remarkable EE compared to the ``SE" optimization case. To delve deeper into the EE gains of our proposed scheme, we examine the changes of SE and energy consumption during convergence in Figs. \ref{fig:Spectrum efficiency} and \ref{fig:Energy Consumption}, respectively. Among all benchmarks, the ``No optimization'' case utilizes all available power but achieves the lowest SE, resulting in the lowest EE. The ``SE'' case achieves the highest SE, twice more than that of the other EE optimization cases. However, to support this high SE, it consumes almost all the allocated energy, approximately eight times more than the ``Passive EE'' case and $100$ times more than the ``Active EE'' case. In comparison to the ``Passive EE'' case, the ``Active EE'' case achieves comparable SE while consuming only $10\%$ of the energy expended in the ``Passive EE'' case. This is attributed to the active RIS's ability to counteract the multiplicative fading effect and reduce energy dissipation during propagation.

\begin{figure}[!t]
    \centering
    \includegraphics[scale=0.45]{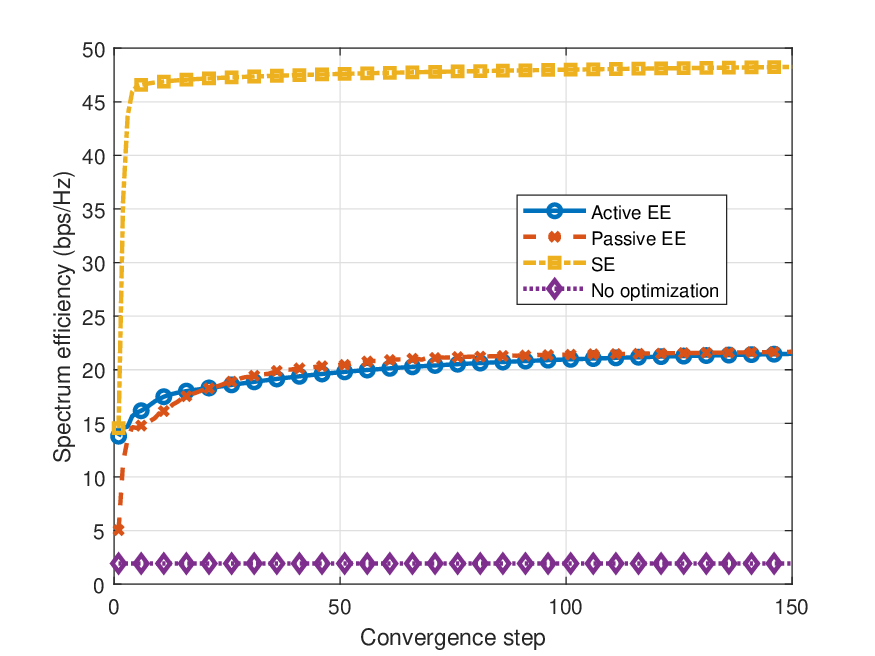}
    \caption{SE performance change during the convergence process: SE versus the convergence step.}
    \label{fig:Spectrum efficiency}
\end{figure}

\begin{figure}[!t]
    \centering
    \includegraphics[scale=0.45]{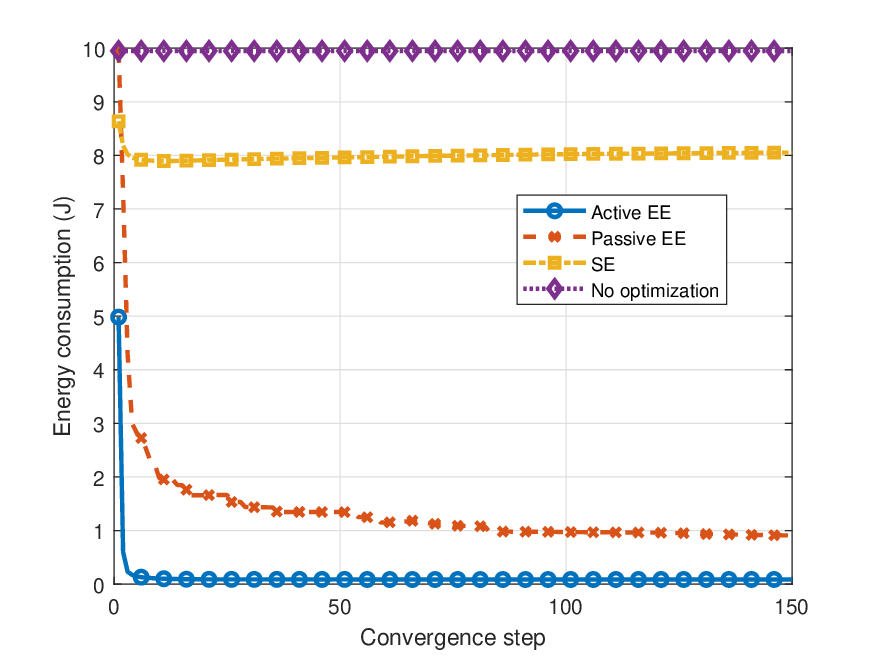}
    \caption{Energy consumption change  during the convergence process: Energy consumption versus the convergence step.}
    \label{fig:Energy Consumption}
\end{figure}
In Fig. \ref{fig:N}, we investigate the relationship between the number of active RIS elements $N$ and the EE of the ISAC system. The number of active RIS elements is varied from $10$ to $80$. As observed in Fig.~\ref{fig:N}, the EE of the ``Active EE" case initially increases as $N$ increases, reaching to a peak, and then decreases. This behavior can be attributed to the interplay between SE and power consumption. When $N$ is small, the increase of $N$ leads to a significant improvement of SE, thereby enhancing the overall EE. However, as $N$ continues to increase, the energy consumed by each RIS element for controlling amplification and phase shift accumulates, resulting in a dramatic increase in the total power consumed by the active RIS. Consequently, the total consumed power outpaces the growth of SE, leading to a decline in EE. The benchmarks exhibit similar trends to the ``Active EE" case due to the same underlying mechanism. On the other hand, the comparison between the ``Active EE" case and the benchmarks reveals that the introduction of active RIS elements in the ISAC system can improve EE, corroborating the findings of previous studies. Additionally, it can be noted that the EE of the ``Active EE" case decreases at a faster rate compared to the ``Passive EE" case. This is understandable because the active RIS must control both the amplifier and the phase tuner, causing its energy consumption to increase more rapidly than that of the passive RIS as $N$ increases. 
\begin{figure}[!t]
    \centering
    \includegraphics[scale=0.45]{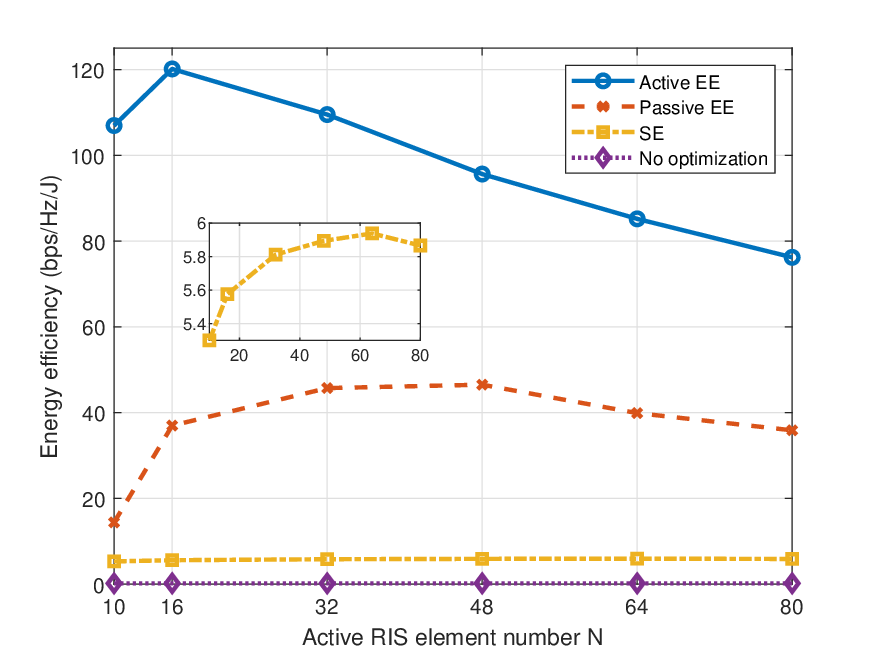}
    \caption{{The effect of the RIS element number $N$ on the EE: EE  versus RIS element number.}}
    \label{fig:N}
\end{figure}

Fig.~\ref{fig:tauk_effect} illustrates the effect of different user SINR thresholds $\tau_k$ on the EE of the ISAC system. The simulated $\tau_k$ values range from $10$ to $50$ in this figure{\footnote{{The user SINR thresholds $\tau_k$ in Fig. \ref{fig:tauk_effect} as well as the echo SNR threshold $\tau_r$ in Fig. \ref{fig:echo SNR} are measured in a linear scale, which are unitless.}}}. As shown in the figure, the EE decreases as $\tau_k$ increases for both ``Active EE" and ``Passive EE" cases. In contrast, the changes of $\tau_k$ have minimal effects on the ``SE" and ``No optimization" cases. This phenomenon can be explained as follows. In both ``Active EE" and ``Passive EE" cases, the goal is to achieve higher SE while minimizing energy consumption under the given constraints. When $\tau_k$ increases, more energy should be consumed to satisfy the stricter constraints, which limits the potential improvement of EE. For the ``SE" case, the EE decreases slightly as $\tau_k$ increases. This is because in most cases, the total energy budget is sufficient to optimize the objective and satisfy the constraints. When $\tau_k$ increases, this constraint only affects a few cases where the channel condition of some users is particularly poor. On average, the EE in the ``SE" case decreases with increasing $\tau_k$, but the effect is not significant. As expected, the ``No optimization" case exhibits no noticeable changes of EE when $\tau_k$ varies, as no optimizations are performed. Furthermore, it can be observed that the active RIS-aided ISAC system can achieve higher EE compared to the passive RIS-aided ISAC system, aligning with the findings of previous experiments. Additionally, using EE as the objective goal instead of SE leads to a more economical scheme for saving energy consumption.

\begin{figure}[!t]
    \centering
    \includegraphics[scale=0.45]{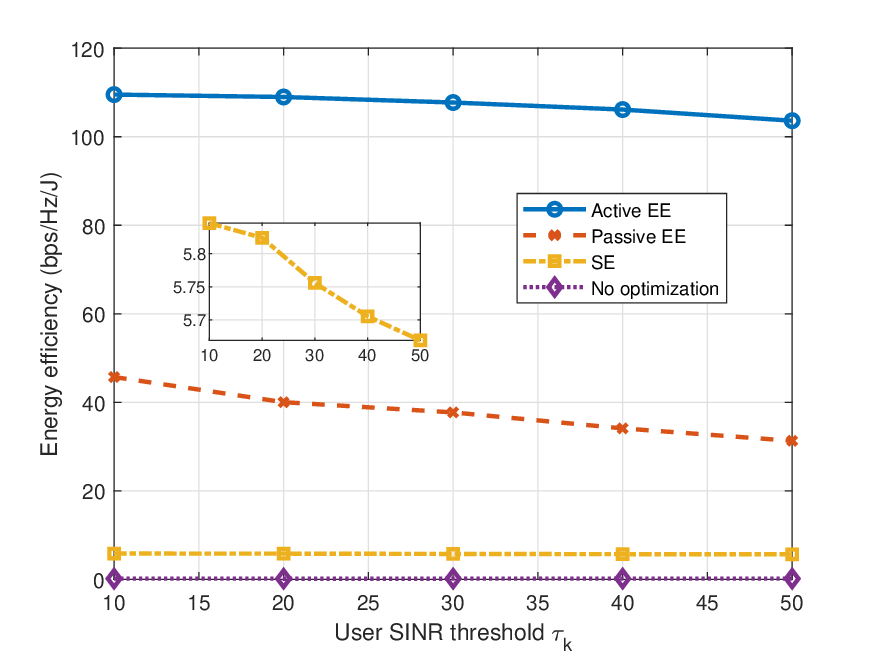}
    \caption{{The effect of the user thresholds $\tau_k$ ($10, 20, 30, 40, 50$)  on EE: EE  versus users' thresholds.}}
    \label{fig:tauk_effect}
\end{figure}

{The impact of different echo SNR thresholds ($\tau_r$) on EE is illustrated in Fig. \ref{fig:echo SNR}, where $\tau_r$ varies in  ($0.5, 1, 5, 10, 15, 20, 25$). This intends to explore the relationships between EE and the sensing performance requirement.} It is observed from the figure that the EE of all scenarios deteriorates with increasing echo SNR requirements ($\tau_r$). This is understandable given the limitations of spectrum and energy resources. When target sensing demands become more stringent, SE decreases under the same energy constraints. However, as expected, the active RIS-assisted ISAC scenario achieves superior EE compared to other benchmarks. Additionally, the ``Active EE'' case exhibits a slower decline compared to the ``Passive EE'' case. This is because active RIS mitigates the multiplicative fading issues, facilitating the echo SNR and user SINR to reach to the predefined thresholds more readily compared to the passive RIS case, particularly for the target echo that undergoes double amplification. Consequently, the echo SNR constraint exerts a lesser influence on the ``Active EE'' case than the ``Passive EE'' case, rendering the proposed approach less sensitive to echo SNR thresholds. {It is also worth mentioning that when the different practical echo SNR ($\gamma_r$) becomes x-axis in the Fig. \ref{fig:echo SNR}, the obtained figure will be similar to that of echo SNR thresholds ($\tau_r$) since the optimized practical echo SNR will approach to the pre-set echo SNR thresholds.}

\begin{figure}[!t]
    \centering
    \includegraphics[scale=0.45]{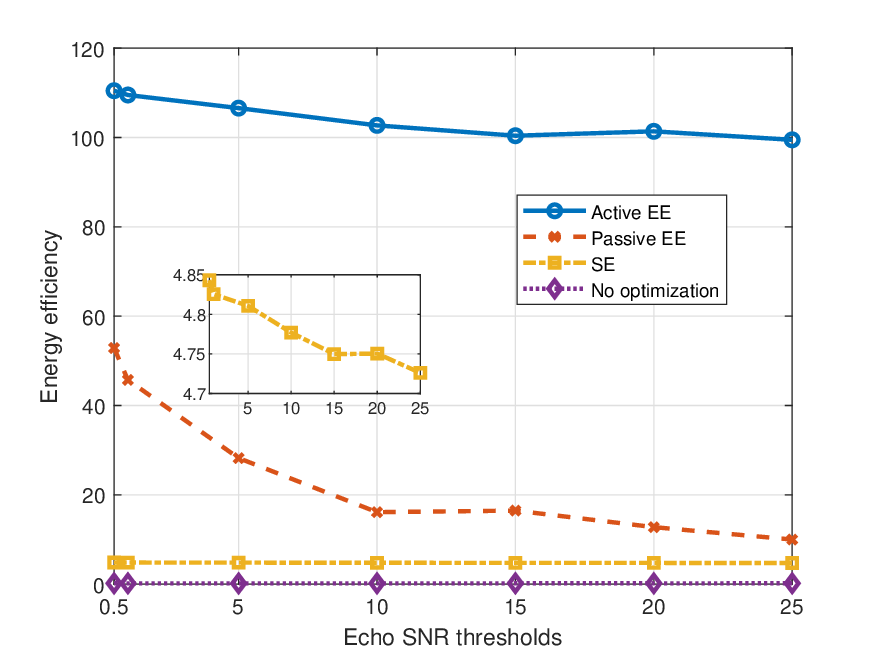}
    \caption{{The effect of different echo SNR thresholds $\tau_r$ ($0.5, 1, 5, 10, 15, 20, 25$) on EE: EE  versus varying echo SNR thresholds.}}
    \label{fig:echo SNR}
\end{figure}

Fig. \ref{fig:amax} illustrates the relationship between EE and the maximum amplification coefficient ($a_{max}$) of active RIS. As $a_{max}$ increases from $2.5$ to $40$, the EE of active RIS-assisted ISAC initially increases before gradually declining. This behavior can be attributed to two primary factors. First, when $a_{max}$ is small and close to $1$, the active RIS essentially reverts to a passive RIS, resulting in EE approaching that of the passive RIS case, as shown in the figure. This is because the energy consumption for active signal processing at the active RIS becomes negligible. Second, for $a_{max}$ values greater than $20$, a gradual decrease of EE is observed. This stems from the increased amplified noise at the active RIS, which contributes to the overall energy consumption. The amplified noise also imposes stricter constraints on the echo SNR and users' SINR, further hindering EE optimization.
 
\begin{figure}[!t]
    \centering
    \includegraphics[scale=0.45]{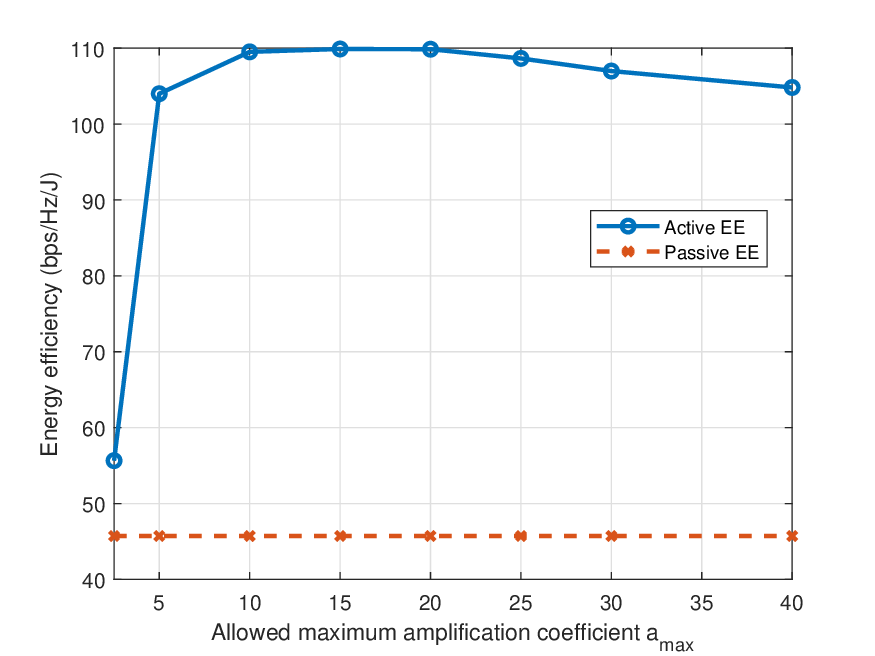}
    \caption{The effect of the maximum amplification coefficient $a_{max}$ ($2.5$, $5$, $10$, $15$, $20$, $25$, $30$, $40$) on the EE: EE versus maximum amplification coefficient.}
    \label{fig:amax}
\end{figure}

\begin{figure}[!t]
    \centering
    \includegraphics[scale=0.45]{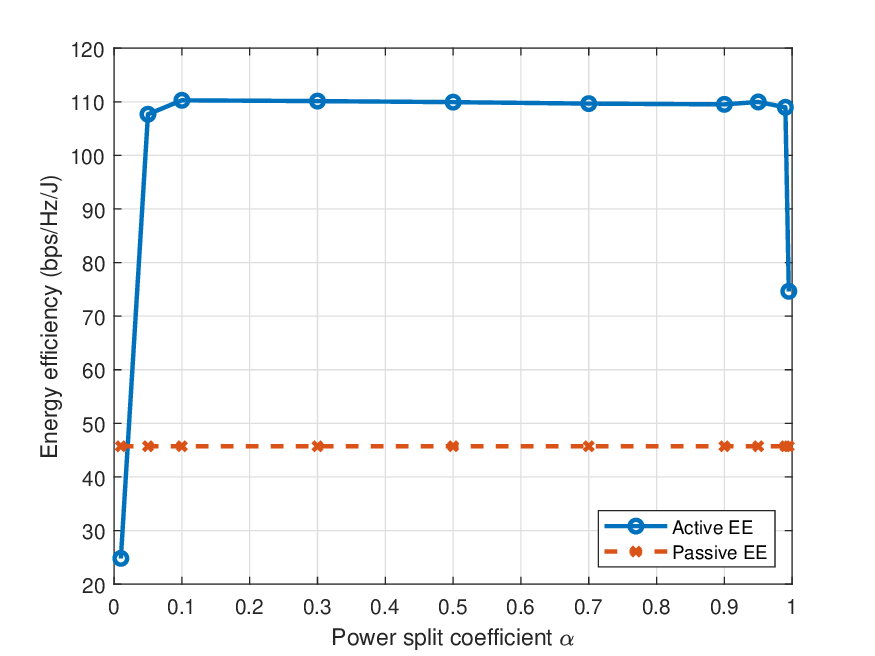}
    \caption{The effect of the power split coefficient $\alpha$ ($0.01$, $0.05$, $0.1$, $0.3$, $0.5$, $0.7$, $0.9$, $0.95$, $0.99$) on the EE: EE versus power split coefficient.}
    \label{fig:alpha}
\end{figure}

The power split between the BS and the active RIS plays a crucial role in determining the EE of the system. Here, we denote $\alpha$ as a power split coefficient of the total power budget, where $\alpha P$ of the power is allocated to the BS while $(1-\alpha) P$ of the power is used in the active RIS. As illustrated in Fig. \ref{fig:alpha}, the EE exhibits a non-monotonic relationship with the power split coefficient $\alpha$. When $\alpha \rightarrow 0$, indicating that most of the energy is allocated to the active RIS, the EE suffers and degrades to a low level. In contrast, when $\alpha$ is neither too large nor too small, the EE stabilizes at around $110$ bps/Hz/J. Finally, as $\alpha \rightarrow 1$, the EE is decreasing and approaching that of the passive RIS case. The observed behaviors can be explained in three phases. \textbf{Phase I where $\alpha \rightarrow 0$: } With a majority of energy allocated to the active RIS, the BS has limited resources to transmit the waveform, leading to a weak transmit signal. Despite the active RIS's amplification capabilities, it cannot fully compensate for the reduced transmission power due to the constrained amplification, resulting in low SE and degraded EE. \textbf{Phase II where $0 < \alpha < 1$: } In this phase, the EE remains relatively stable around 110 bps/Hz/J. This stability stems from the minimization of actual consumed energy during the optimization process, as shown in Fig. \ref{fig:Energy Consumption}. Regardless of the power split coefficient between $0.1$ and $0.9$, the allocated energy becomes redundant, allowing the optimization algorithm to find a solution with consistently low energy consumption. As a result, the optimization process minimizes the actual consumed energy, rendering the initial energy allocation less significant. \textbf{Phase III where $\alpha \rightarrow 1$: } As the power split coefficient approaches 1, RIS behaves more like a passive RIS, leading to a decline of EE. This decrease is attributed to the reduced amplification capabilities of passive RIS compared to active RIS. To summarize, the power split coefficient significantly affects the system EE. While a moderate power split can lead to optimal EE, allocating too much or too little energy to either the BS or the active RIS can result in performance degradation.

\section{Conclusions}\label{conclusions}
In this paper, we explore the EE enhancement of an ISAC system assisted by an active RIS. Through joint optimization of the transceiver beamformer at BS, the amplification matrix, and the reflecting matrix of  active RIS, we maximize system's EE while adhering to power constraints, SINR requirements for both users, and the target echo signal. To solve the optimization problem, an effective alternative optimization algorithm is proposed utilizing the generalized Rayleigh quotient optimization approach, SDR, and MM framework. Simulation results validate that the active RIS-assisted ISAC system achieves significant EE improvement compared to both passive RIS case and SE optimization case.

\begin{appendices}        
\section{PROOF OF Lemma \ref{lm-1}}  \label{app_a}
Consider any Hermitian matrix $\mathbf{A}$ with a maximum eigenvalue $\lambda$. Due to the non-negative nature of all its eigenvalues, $\lambda \mathbf{I}\! -\! \mathbf{A}$ is a positive semidefinite matrix. Consequently, $f(\mathbf{x}) \!= \!\mathbf{x}^H (\lambda \mathbf{I} \!- \! \mathbf{A}) \mathbf{x}$ is a convex function, implying that it exceeds its first-order Taylor expansion according to the first-order condition of convex functions. This can be expressed mathematically as
\begin{align} \label{first-order condition}
    \!\!\mathbf{x}^H (\lambda\mathbf{I}\!-\!\mathbf{A})\mathbf{x}&\!\geq \!\mathbf{x}_s^H (\lambda\mathbf{I}\!-\!\mathbf{A})\mathbf{x}_s\!+ \!(\mathbf{x}\!-\!\mathbf{x}_s)^H\frac{\partial \mathbf{x}^H (\lambda\mathbf{I}\!-\!\mathbf{A})\mathbf{x}}{\partial \mathbf{x}^*}|_{\mathbf{x}\!=\!\mathbf{x}_s} \notag \\
    &+\frac{\partial \mathbf{x}^H (\lambda\mathbf{I}-\mathbf{A})\mathbf{x}}{\partial \mathbf{x}^T}|_{\mathbf{x}=\mathbf{x}_s} (\mathbf{x}-\mathbf{x}_s).
\end{align}
Since
\begin{align}
    &\frac{\partial \mathbf{x}^H (\lambda\mathbf{I}-\mathbf{A})\mathbf{x}}{\partial \mathbf{x}^*}|_{\mathbf{x}=\mathbf{x}_s}=(\lambda\mathbf{I}-\mathbf{A})\mathbf{x}_s, \\
    & \frac{\partial \mathbf{x}^H (\lambda\mathbf{I}-\mathbf{A})\mathbf{x}}{\partial \mathbf{x}^T}|_{\mathbf{x}=\mathbf{x}_s}= \mathbf{x}_s^H(\lambda\mathbf{I}-\mathbf{A})^H,
\end{align}
the inequality (\ref{first-order condition}) can then be written as
\begin{align}
    &\mathbf{x}^H (\lambda\mathbf{I}-\mathbf{A})\mathbf{x}\geq \mathbf{x}_s^H (\lambda\mathbf{I}-\mathbf{A})\mathbf{x}_s \notag\\
    &+(\mathbf{x}-\mathbf{x}_s)^H(\lambda\mathbf{I}-\mathbf{A})\mathbf{x}_s+\mathbf{x}_s^H(\lambda\mathbf{I}-\mathbf{A})(\mathbf{x}-\mathbf{x}_s),
\end{align}
which can be rearranged as
\begin{align}
    \!\!\!\mathbf{x}^H \mathbf{A} \mathbf{x} \!\leq \!\lambda \mathbf{x}^H \mathbf{x} \!+\! \mathbf{x}_s^H (\lambda\mathbf{I}\!-\!\mathbf{A})\mathbf{x}_s \!+\! 2 \mathfrak{R}\!\left\lbrace \mathbf{x}^H(\mathbf{A}\!-\!\lambda\mathbf{I})\mathbf{x}_s\!\right\rbrace.
\end{align}
This ends the proof.

\section{PROOF OF Proposition \ref{prop-1}}  \label{app_b}
The relationship between the maximum eigenvalue of $\lambda_{v1}$ associated with $\mathbf{V}_1$ and the maximum eigenvalue of $\lambda_{v2}$ associated with $\mathbf{V}_2$ can be readily established by considering equation (\ref{nota4}{c}), which gives
\begin{align}
     \lambda_{v1} &= \mu \zeta \lambda_{v2}.
\end{align}

To determine the eigenvalue $\lambda_{v2}$, we first observe $\mathbf{V}_2 \!=\! \overline{\mathbf{H}}_{rt}^T \otimes \overline{\mathbf{W}}_t \!+\! \sigma_0^2 (\overline{\mathbf{H}}_{rt}^T \!\otimes \!\overline{\mathbf{H}}_{rt})$. Taking advantage of two properties of the Kronecker operator, i.e. $k(\mathbf{X}\otimes\!\mathbf{Y})\!=\!(k\mathbf{X})\otimes\mathbf{Y}\!=\!\mathbf{X}\otimes\! (k\mathbf{Y})$ and $\mathbf{X}\!\otimes\!(\mathbf{Y}\!\pm\!\mathbf{Z})\!=\!\mathbf{X}\!\otimes\!\mathbf{Y}\!\pm\!\mathbf{X}\!\otimes\!\mathbf{Z}$, we can re-write $\mathbf{V}_2$ as
\begin{align}
\mathbf{V}_2 = \overline{\mathbf{H}}_{rt}^T \otimes (\overline{\mathbf{W}}_t + \sigma_0^2\overline{\mathbf{H}}_{rt}).
\end{align}
Leveraging the eigenvalue property of the Kronecker operator, it is established that $\lambda_{v2}$, the maximum eigenvalue of $\mathbf{V}_2$, is the product of the maximum eigenvalue of $\overline{\mathbf{H}}_{rt}$ and the maximum eigenvalue of $(\overline{\mathbf{W}}_t \!+\! \sigma_0^2\overline{\mathbf{H}}_{rt})$. Since $\overline{\mathbf{H}}_{rt} \!=\! \operatorname{diag}(\mathbf{h}_{rt}^H) \operatorname{diag}(\mathbf{h}_{rt})$ is a real diagonal matrix, its maximum eigenvalue is simply the maximum value of its diagonal elements, denoted as $\max\{\overline{\mathbf{H}}_{rt}\}$. Regarding $(\overline{\mathbf{W}}_t \!+ \!\sigma_0^2\overline{\mathbf{H}}_{rt})$, considering that $\|\mathbf{W}\|_F^2 \!\gg\! \sigma_0^2$, the elements in $\overline{\mathbf{W}}_t$ are significantly larger than those in $\sigma_0^2\overline{\mathbf{H}}_{rt}$. Consequently, the maximum eigenvalue of $(\overline{\mathbf{W}}_t \!+\! \sigma_0^2\overline{\mathbf{H}}_{rt})$ can be approximated with that of $\overline{\mathbf{W}}_t$, denoted as $\lambda_w$. However, determining the eigenvalue of $\overline{\mathbf{W}}_t$ remains computationally challenging, especially when $N$ is large. Upon closer examination, it can be observed that $\overline{\mathbf{W}}_t$ is rank-deficient since $N \!>\! (M+K) \!>\! M$ in general. To leverage this characteristic, we define $\overline{\mathbf{W}}_{t2} \!= \!\mathbf{W}^H \mathbf{G}^H \operatorname{diag}(\mathbf{h}_{rt}) \operatorname{diag}(\mathbf{h}_{rt}^H) \mathbf{G} \mathbf{W} \!\in \!\mathbb{C}^{(M+K)\times (M+K)}$, which possesses a lower dimension compared to $\overline{\mathbf{W}}_t$ and shares the non-zero eigenvalues with $\overline{\mathbf{W}}_{t}$. Therefore, the maximum eigenvalue $\lambda_w$ of $\overline{\mathbf{W}}_{t2}$ can be efficiently determined by performing eigenvalue decomposition of the lower-dimensional matrix $\overline{\mathbf{W}}_{t2}$. Consequently, the maximum eigenvalue of $\mathbf{V}_2$ can be approximated as
 \begin{align}
     \lambda_{v2} &\approx \lambda_w \max\lbrace \overline{\mathbf{H}}_{rt} \rbrace.
 \end{align}

As for the eigenvalue $\lambda_{v3}$, we first rewrite $\mathbf{V}_3$ as
\begin{align}
    \mathbf{V}_3= (\tau_r \sigma_0^2\overline{\mathbf{H}}_{rt} - \overline{\mathbf{W}}_{t})^T \otimes \overline{\mathbf{U}}_t.
\end{align}
To ascertain the maximum eigenvalue of $\mathbf{V}_3$, we need to determine the maximum eigenvalues of $(\tau_r \sigma_0^2\overline{\mathbf{H}}_{rt}^T \!- \!\overline{\mathbf{W}}_{t}^T)$ and $\overline{\mathbf{U}}_t$, respectively. Firstly, we observe that $\overline{\mathbf{U}}_t \!=\! \mathrm{diag}(\mathbf{h}_{rt}^H) \mathbf{G} \mathbf{u} \mathbf{u}^H \mathbf{G}^H \mathrm{diag}(\mathbf{h}_{rt})$, which clearly reveals that $\overline{\mathbf{U}}_t$ is a rank-one matrix with only one non-zero eigenvalue. Since $\mathrm{tr}(\overline{\mathbf{U}}_t) \!\!= \!\! \sum_{n=1}^{N} \!\lambda_{u,n}$, where $\lambda_{u,n}$ represents all eigenvalues of $\overline{\mathbf{U}}_t$, the maximum eigenvalue of $\overline{\mathbf{U}}_t$ is equal to $\mathrm{tr}(\overline{\mathbf{U}}_t)$. Regarding $\tau_r \sigma_0^2\overline{\mathbf{H}}_{rt} \!- \!\overline{\mathbf{W}}_{t}$, it is a full-rank matrix, implying that all its eigenvalues are non-zero. Consequently, its maximum eigenvalue, denoted as $\lambda_t$, can only be determined through direct eigenvalue decomposition.

Consequently, the eigenvalue $\lambda_{v3}$ can be written as the product of maximum eigenvalue of $\overline{\mathbf{U}}_t$ and maximum eigenvalue of $\overline{\mathbf{H}}_{rt}$, given that 
\begin{align}
    \lambda_{v3} &= \lambda_{t} \rm\textbf{\rm tr}(\overline{\mathbf{U}}_t) .
\end{align}
This ends the proof. 

\end{appendices} 
\bibliographystyle{IEEEtran}
\bibliography{references}{}

\end{document}